\documentclass{imsart}

\RequirePackage{amsthm,amsmath,amsfonts,amssymb}
\RequirePackage[authoryear]{natbib}
\RequirePackage[colorlinks,citecolor=blue,urlcolor=blue]{hyperref}
\RequirePackage{graphicx,subfig}
\usepackage{xcolor}

\startlocaldefs
\theoremstyle{plain}

\newtheorem{theorem}{Theorem}[section]
\newtheorem{lemma}[theorem]{Lemma}
\theoremstyle{definition}

\newtheorem{remark}{Remark}
\providecommand{\tabularnewline}{\\}

\DeclareMathOperator{\V}{\mathbb{V}}
\DeclareMathOperator{\E}{\mathbb{E}}
\DeclareMathOperator{\Cov}{Cov}

\DeclareMathOperator{\eigmax}{eigmax}
\DeclareMathOperator{\eigmin}{eigmin}

\endlocaldefs

\begin{document}

\begin{frontmatter}
\title{Ridge Estimation of High Dimensional Two-Way Fixed Effect Regression}
\runtitle{Ridge Estimation of High Dimensional Two-Way Fixed Effect Regression/}

\begin{aug}
\author[A]{\fnms{Junnan}~\snm{He} \ead[label=e1]{junnan.he@sciencespo.fr}}
\and
\author[A]{\fnms{Jean-Marc}~\snm{Robin}\ead[label=e2]{jeanmarc.robin@sciencespo.fr}}
\address[A]{Economics Department, Sciences Po, Paris \printead[presep={ ,\ }]{e1,e2}}
\end{aug}

\begin{abstract}
We study a ridge estimator for the high-dimensional two-way fixed
effect regression model with a sparse bipartite network. 
We develop concentration inequalities showing
that when the ridge parameters increase as the log of the network
size, the bias and the variance-covariance matrix of the vector of
estimated fixed effects converge to deterministic equivalents that
depend only on the expected network. We provide simulations and an
application using administrative data on wages for worker-firm matches.

\end{abstract}



\end{frontmatter}

\section{Introduction}

In this paper, we study the two-way fixed-effect model: 
\[
y_{ijs}=\mu_{i}+\phi_{j}+u_{ijs},i\in\{1,...,n\},j\in\{1,...,p\},s\in\{0,...,d_{ij}\},
\]
when the observed matches $(i,j)$ have a sparse bipartite graph structure.
In our running example, $y_{ijt}$ is the log of the annualized earnings
of worker $i$ employed by firm $j$ in the $s$th occurrence of that
particular match, or the residual from a preliminary regression on observed covariates
 (\cite{AbowdKramarzMargolis1999}). In most applications, the
worker panel covers a limited number of years (5 or 7 years). 
The network of firm-worker matches is therefore sparse as workers sample
a small number of firms (i.e. $d_{ij}=0$ for most $j$) during the
observation period. The number of wage observations per worker ($d_{i\cdot}=\sum_{j}d_{ij}$)
is of the order of the number of years of observations, sometimes
lower because of unemployment, sometimes greater because of infra-year
job mobility.  
Similar two-way fixed effect structures arise in a variety of settings 
(see \cite{LarremoreClausetJacobs2014}). 
For example, if $i$ indexes authors and $j$ indexes journals (or fields), one can model a paper-level outcome such as citations as the sum of an author effect and an outlet effect when authors publish
a small number of papers in a fixed time interval. Likewise, if $i$ indexes pollinator species and $j$ indexes plant species, an interaction-level measure such as fruit set (or seed set) can be modeled as the sum of a pollinator ``effectiveness” component and a plant ``receptivity” component. For ease of exposition, we refer to $y$ as the wage,  $i$ as workers, and $j$ as firms in what follows. 

Our parameter of interest is the distribution of fixed effects. Obviously, with fixed or bounded degrees
no fixed effect can be consistently identified, but what about their distribution? Can one say something
on the limit of the distribution of estimated fixed effects when the size of the network becomes large? 
Obviously, the answer to this question should depend on the estimator.

If the network is sparse, the OLS estimator of the fixed-effect model
faces a weak identification problem. \cite{JochmansWeidner2019}
show that the variance of the OLS estimator depends on the value of
the second lowest eigenvalue of the Laplacian, which determines the connectivity of the network. The Cheeger inequality
(\cite{Chung1997}) bounds the second eigenvalue of the Laplacian
and the Cheeger constant measures the existence of a ``bottleneck''.
\cite{ZhangRohe2018} further show that sparse, inhomogeneous graphs
(edges are drawn independently with node-dependent probabilities)
generate a number of bottlenecks (``dangling nodes'') that increases
with the number of nodes. The number of very small eigenvalues of
the Laplacian increases correspondingly.

This issue is well studied in the community detection literature,
aiming at finding subsets of nodes with a higher probability to connect
to each other. For sparse graphs, node clustering
often returns a partition with one big cluster containing most of
the data and many small clusters. \cite{ChaudhuriGrahamTsiatas2012}
proposed to regularize the graph by making it less sparse by adding
many ``weak links''. Concretely, this amounts to replacing the adjacency
matrix $A$ by $A+\lambda\mathbf{1}_{n}\mathbf{1}_{n}^{\top}$, for
a small $\lambda$. Since then, several other ways of regularizing
the graph have been proposed. The solution of \cite{QinRohe2013}
is of particular interest to us because it relates to the ridge regression
estimator. They only add a small number to the degree matrix, not
to the whole adjacency matrix. Recent work by \cite{DallAmicoCouilletTremblay2021}
shows that Qin and Rohe's particular form of regularization (as well
as the equivalent one Dall'Amico et al. propose) outperforms other
competing algorithms. In this paper, we will build on the asymptotic
theory developed in the statistical network literature to advance
our understanding of the asymptotic properties of the ridge estimator
of the two-way fixed effect model with a sparse bipartite network structure.

Three features of the underlying network are important to incorporate
in this analysis. First, the network is large (many workers and firms)
and so is the number of fixed effects that we seek to estimate. Second,
the network is sparse, as in all practical situations only a small
fraction of worker-firm matches are observed out of the total number
of possible ones. Third, the asymptotic theory will assume a large
number of workers $n$ and firms $p$, but $p/n$ is largely fixed
and independent of $n$. This implies that the degrees of worker and
firm nodes are more or less fixed and independent of the size of the
network. 

Ridge regression in high-dimensional setups is a well-studied problem
in the statistical literature. The literature focuses on the regression model
$y_{i}=x_{i}^{\top}\beta+u_{i}$, with $n$ independent observations
and where the number of regressors $p$ grows with the sample size
with $p/n\rightarrow\gamma$. The vector of regressors is assumed
spherical, i.e. $x_{i}=\Sigma^{1/2}z_{i}$ where $z_{i}$ has independent
entries. The vector of parameters $\beta$ is random with mean 0 and
variance $p^{-1}\alpha^{2}I_{p}$. The ridge regression estimator
is then $\widehat{\beta}=\left(X^{\top}X+\lambda I_{p}\right)^{-1}X^{\top}Y$,
for $Y=(y_{i})$ and $X^{\top}=(x_{1}...x_{n})$. Under these assumptions,
\cite{DobribanWager2018} show that the expected predictive risk
converges almost surely to a deterministic limit that is minimized
by the optimal parameter $\lambda^{*}=n\gamma\alpha^{-2}\sigma^{2}$,
where $\sigma^{2}$ is the variance of residuals.\footnote{See \cite{Dicker2013,ElKaroui2018,RichardsMourtadaRosasco2020,WuXu2020,BigotDaboMale2024,PatilDuTibshirani2025}
for other references.} This setup is rather different from ours. 
In our setup the elements of $\beta$ do not shrink to zero. At the same time the linear
form $x_{i}^{\top}\beta$ does not diverge because of sparsity:
the number of nonzero entries in $x_{i}$  increases with $n$
and $p$. 

We will need a model for the underlying network. We assume a Degree-Corrected
Stochastic Block model (\cite{KarrerNewman2011,LarremoreClausetJacobs2014}),
which has become the standard model in the community detection literature.
Specifically, we first assume that the set of workers and firms can
be partitioned into $K$ groups (communities). Then, the number $d_{ij}$
of links between $i$ and $j$ is drawn independently from a Binomial$(T,p_{ij})$, where $p_{ij}$ is the probability
of a link. This probability is higher if $i$ and $j$ belong to the
same community.\footnote{See also \cite{Nimczik2017,AbowdMcKinneySchmutte2019,BonhommeLamadonManresa2019,LentzPiyapromdeeRobin2023}.
These papers usually allow for a Markovian process of job mobility, but
a little dependence between $d_{ij}$ and $d_{ij'}$ should not
affect the overall stochastic structure of the network.}

For any formed network of worker and firm links, we then generate
a two-way fixed effect model with homoscedastic residuals. Let $\beta=(\mu^{\top},\phi^{\top})^{\top}$
denote the vector of parameters. We consider the ridge estimator $\widehat{\beta}$
with $X=(W,F)$, where $W$ is the matrix of worker dummies and $F$
that of firm dummies, using two different regularization parameters,
one for $\mu$ and one for $\phi$. 
Our main theoretical contribution is to derive deterministic equivalents for the ridge estimator’s first- and second-order errors. Under our random graph model, if the ridge penalties
are of order $\ln(n+p)$, then the high dimensional bias vector $\mathbb E[\widehat{\beta}-\beta]$
and its variance matrix are close  to deterministic limits in operator matrix norm. These deterministic
equivalents are computed from the expected network, where each node has the
expected degree and where the adjacency matrix is the expected adjacency
matrix. 
This  helps explain why the overall distribution of fixed effects can be accurately approximated even though each fixed effect cannot be consistently estimated because the number of observations per worker and firm does not increase with the network size. A further study of these deterministic equivalents using random matrix theory is beyond the scope of this paper. 

We run a set of Monte Carlo experiments showing that ridge regularization substantially improves the recovery of the fixed-effect distribution relative to OLS. Moreover, the  $\ln(n+p)$ scaling is a sufficient condition; our simulations indicate that tuning ridge penalties using prediction error (cross-validation) tends to yield rather lower ridge parameters.
Following the empirical literature, we consider the decomposition
of the variance of the outcome variance into the contributions of
the variance of the worker effect, the variance of the firm effect,
the contribution of the covariance between worker and firm effects,
and the variance of residuals. When the network is very sparse, homoscedastic
bias correction fails,\footnote{See \cite{AndrewsGillSchankUpward2008}, \cite{Gaure2014}, \cite{KlineSaggioSoelvsten2020}
and \cite{AzkarateAskasuaZerecero2022}.} while the ridge estimator yields more accurate variance decompositions.
We also report estimations on actual matched employer-employee data
(the French DADS panel). The log wage variance decomposition derived
from OLS estimates using the largest connected component delivers
a negative contribution of the fixed effect covariance. We find that
homoscedastic bias corrections are minimal. However, the ridge estimator
yields a more reasonable variance decomposition. The worker effect
explains a more reasonable share of the total variance of log wages
and the covariance term is positive instead of being negative. We
also consider an estimation with the largest strongest connected component,
which allows to use the Leave-One-Out (LOO) bias correction method
of \cite{KlineSaggioSoelvsten2020}. OLS and the homoscedastic bias
correction improve, but very little. Ridge does not change much and
the LOO bias correction comes close to ridge.

The layout of the paper is as follows. In the next section, we describe
the ridge estimator and relate it to the underlying graph of connections
between workers and firms. Section \ref{sec:SBM} develops a Stochastic
Block Model for the network. Asymptotic theory is developed in Section
\ref{sec:asymptotics_network} (the network) and Section \ref{sec:asymptotics_ridge}
(the ridge estimator). The proofs are in the Appendix.

\section{The two-way fixed effect model}

\label{sec:2way_model}

The two-way fixed effect model can be written in matrix form as 
\[
Y=W\mu+F\phi+U=X\beta+U,\quad X=(W,F),\quad\beta=(\mu^{\top},\phi^{\top})^{\top}.
\]
The vector $Y=(y_{ijs})\in\mathbb{R}^{N}$ (with $N=\sum_{i,j}d_{ij}$)
is a vector of continuous outcomes. 
The matrix $W\in\{0,1\}^{N\times n}$ associates outcome observations
to worker IDs $i\in\{1,...,n\}$ and the matrix $F\in\{0,1\}^{N\times p}$
associates observations to firm IDs $j\in\{1,...,p\}$. We assume
that $U$ is iid with mean 0 and variance $\sigma^{2}I_{n}$.

\subsection{OLS and bipartite graph interpretation}\label{subsec:OLS}

The OLS estimator $\widehat{\beta}=(\widehat{\mu}^{\top},\widehat{\phi}^{\top})^{\top}$
solves the normal equations $X^{\top}X\widehat{\beta}=X^{\top}Y$
with 
\begin{align*}
X^{\top}X=\left(\begin{array}{cc}
W^{\top}W & W^{\top}F\\
F^{\top}W & F^{\top}F
\end{array}\right):= & \left(\begin{array}{cc}
D_{w} & B\\
B^{\top} & D_{f}
\end{array}\right),
\end{align*}
and 
\[
\det\left(X^{\top}X\right)=\det\left(D_{w}\right)\det\left(D_{f}-B^{\top}D_{w}^{-1}B\right)=\det\left(D_{f}\right)\det\left(D_{w}-BD_{f}^{-1}B^{\top}\right).
\]
Each entry of $B=W^{\top}F=(d_{ij})$ is the number of wage observations
per match $(i,j)$. It is a weighted adjacency matrix of the bipartite
graph connecting workers and firms. We focus on the case where  $B$ is sparse: $d_{ij}=0$ for most potential links. We also
have 
\[
D_{w}=W^{\top}W=\text{diag}(d_{i\cdot}),\quad D_{f}=F^{\top}F=\text{diag}(d_{\cdot j}).
\]
These are the degree matrices of workers and firms: $d_{i\cdot}=\sum_{j}d_{ij}$
is the number of observations for worker $i$, and $d_{\cdot j}=\sum_{i}d_{ij}$
is the number of observations for firm $j$. We assume that worker
and firm degrees are all positive, ensuring that $\det\left(D_{w}\right)\ne0$
and $\det\left(D_{f}\right)\ne0$.

Assuming that no worker or firm has zero degree, the matrix 
\[
\widetilde{A}_{f}=B^{\top}D_{w}^{-1}B=\left(\sum_{i}\frac{d_{ij}d_{ij'}}{d_{i\cdot}}\right)_{j,j'}.
\]
is a weighted adjacency matrix of the undirected graph connecting
firms through common employees, also called a weighted one-mode projection
of the original bipartite graph (see \cite{ZhouRenMedoZhang2007}).
This matrix is sparse if for most firm pairs $(j,j')$, no worker is observed at both firms over the observation window.
In practice, for the usual lengths of periods of observation
(5 or 7 years), this is typically the case.
The matrix $\widetilde{L}_{f}:=D_{f}-B^{\top}D_{w}^{-1}B=F^{\top}M_{W}F$,
where $M_{W}=I_{N}-W(W^{\top}W)^{-1}W^{\top}$, is the corresponding
Laplacian. 

A similar interpretation applies to the matrix 
\[
\widetilde{A}_{w}=BD_{f}^{-1}B^{\top}=\left(\sum_{j}\frac{d_{ij}d_{i'j}}{d_{\cdot j}}\right)_{i,i'}.
\]
It is the adjacency matrix of the undirected, weighted graph connecting
workers through common employers. The matrix $\widetilde{L}_{w}:=D_{w}-BD_{f}^{-1}B^{\top}=W^{\top}M_{f}W$
is the worker Laplacian.

In the sequel, we shall use the following normalized, symmetric versions
of these adjacency matrices: 
\[
A_{w}=\left(\sum_{j}\frac{d_{ij}d_{i'j}}{\sqrt{d_{i\cdot}d_{i'\cdot}}d_{\cdot j}}\right)_{i,i'}=EE^{\top},\quad A_{f}=\left(\sum_{i}\frac{d_{ij}d_{ij'}}{\sqrt{d_{\cdot j}d_{\cdot j'}}d_{i\cdot}}\right)_{j,j'}=E^{\top}E,
\]
with 
\[
E=D_{w}^{-1/2}BD_{f}^{-1/2}=\left(\frac{d_{ij}}{\sqrt{d_{i\cdot}d_{\cdot j}}}\right)_{i,j}.
\]
We also define the normalized versions of the Laplacians: $L_{w}=I_{n}-A_{w}$
and $L_{f}=I_{p}-A_{f}$. The normalized adjacency and the Laplacian
matrices satisfy the following properties. 
\begin{lemma}
\label{lem:Laplacian eig}The eigenvalues of $A_{f}$ (say $\alpha_{p}\le...\le\alpha_{1}$)
and of $L_{f}$ ($\lambda_{1}=1-\alpha_{1}\le...\le\lambda_{p}=1-\alpha_{p}$)
are in $[0,1]$. Moreover, $\lambda_{1}=0$ and $\alpha_{1}=1$. 
\end{lemma}
\begin{proof}
See Appendix. Unless specified otherwise, all proofs are in the appendix.
\end{proof}
Lemma \ref{lem:Laplacian eig} bounds the eigenvalues of the Laplacian
matrices between 0 and 1. It also asserts that whatever the connectedness
of the graph, 0 is always an eigenvalue of the Laplacian. This is
a consequence of the fact that the columns of $W$ and $F$ sum to
one. In practice, this implies that the two-way fixed effect regression
requires a normalization of the fixed effects $\mu$ and $\phi$.
The next lemma identifies the multiplicity of the eigenvalue 0 to
the number of connected components of the graph. 
\begin{lemma}
(i) The different versions of the Laplacian share identical eigenvalues.
(ii) The number of disconnected components is equal to the multiplicity
of the eigenvalue 0 of the Laplacians or the eigenvalue 1 of the adjacency
matrices. 
\end{lemma}
\begin{proof}
See for example \cite{Luxburg2007}. 
\end{proof}
To calculate the OLS estimator, one parameter normalization is required
for each connected component. In order to estimate the two-way fixed
effect model, one must first find connected components, using Tarjan's
algorithm for example. Estimation is usually performed on the biggest
component.

\subsection{The ridge estimator}\label{subsec:ridge-est}

In practice, the biggest connected component may be weakly connected.
That is to say, the second lowest eigenvalue of the Laplacian matrix
may be close to zero. This is a weak identification issue. The
following ridge estimator is a natural solution (\cite{Hastie2020}): 
\begin{align*}
\widehat{\beta}=\left(\begin{array}{c}
\widehat{\mu}\\
\widehat{\phi}
\end{array}\right) & =\left(\begin{array}{cc}
D_{w}+\lambda_{w}I_{n} & B\\
B^{\top} & D_{f}+\lambda_{f}I_{p}
\end{array}\right)^{-1}\left(\begin{array}{c}
W^{\top}Y\\
F^{\top}Y
\end{array}\right)
\end{align*}
where $\lambda_{w},\lambda_{f}$ are nonnegative parameters.
Denote $D_{w,\lambda}=D_{w}+\lambda_{w}I_{n}$ and $D_{f,\lambda}=D_{f}+\lambda_{f}I_{p}$.
The estimator is feasible if 
\[
\det\left(\begin{array}{cc}
D_{w,\lambda} & B\\
B^{\top} & D_{f,\lambda}
\end{array}\right)=\det\left(D_{w,\lambda}\right)\det\left(\widetilde{L}_{f,\lambda}\right)=\det\left(D_{f,\lambda}\right)\det\left(\widetilde{L}_{w,\lambda}\right)\ne0,
\]
where $\widetilde{L}_{w,\lambda}:=D_{w,\lambda}-BD_{f,\lambda}^{-1}B^{\top}$
and $\widetilde{L}_{f,\lambda}:=D_{f,\lambda}-B^{\top}D_{w,\lambda}^{-1}B$
are the regularized Laplacians. In which case, blockwise inversion
yields 
\[
\left(\begin{array}{c}
\widehat{\mu}\\
\widehat{\phi}
\end{array}\right)=\left(\begin{array}{c}
\widetilde{L}_{w,\lambda}^{-1}\left(W^{\top}Y-BD_{f,\lambda}^{-1}F^{\top}Y\right)\\
\widetilde{L}_{f,\lambda}^{-1}\left(F^{\top}Y-B^{\top}D_{w,\lambda}^{-1}W^{\top}Y\right)
\end{array}\right).
\]
Let $L_{w,\lambda}=I_{n}-D_{w,\lambda}^{-1/2}BD_{f,\lambda}^{-1}B^{\top}D_{w,\lambda}^{-1/2}$
and $L_{f,\lambda}=I_{p}-D_{f,\lambda}^{-1/2}B^{\top}D_{w,\lambda}^{-1}BD_{f,\lambda}^{-1/2}$
denote normalized versions of the regularized Laplacians. Then, $\widetilde{L}_{w,\lambda}=D_{w,\lambda}^{1/2}L_{w,\lambda}D_{w,\lambda}^{1/2}$
and $\widetilde{L}_{f,\lambda}=D_{f,\lambda}^{1/2}L_{f,\lambda}D_{f,\lambda}^{1/2}$.

The following lemma shows that the ridge estimator is always feasible (even
without normalizing the fixed effect parameters) by proving that the
eigenvalues of the regularized Laplacian are bounded away from zero.
Let $E_{\lambda}=D_{w,\lambda}^{-1/2}BD_{f,\lambda}^{-1/2}.$ The
regularized firm Laplacian can be written as $L_{f,\lambda}=I_{p}-E_{\lambda}^{\top}E_{\lambda}$,
and the regularized worker Laplacian as $L_{w,\lambda}=I_{n}-E_{\lambda}E_{\lambda}^{\top}$.
Hereafter, the matrix norm is the spectral matrix norm ($\left\Vert A\right\Vert =\sqrt{\eigmax(A^{\top}A)}$
for any matrix $A$). 
\begin{lemma}
\label{lem:max-eig-E'E}If the associated ridge parameter is positive
($\lambda_{f}>0$) the smallest eigenvalue of the firm Laplacian is
strictly positive; specifically: $\eigmin(L_{f,\lambda})\ge\frac{\lambda_{f}}{\max_{j}d_{\cdot j}+\lambda_{f}}>0$.
It follows that $\left\Vert E_{\lambda}\right\Vert =\sqrt{\eigmax(E_{\lambda}^{\top}E_{\lambda})}\le\sqrt{\frac{\max_{j}d_{\cdot j}}{\max_{j}d_{\cdot j}+\lambda_{f}}}<1$.
The case for $\lambda_{w}$ is similar. 
\end{lemma}

\subsection{In-sample bias and variance}\label{subsec:In-sample-bias-and-var}

Suppose that $U$ is independent of $X=(W,F)$, has mean 0 and $\V U=\sigma^{2}I_{N}$.
Then, the ridge estimator is biased with 
\begin{align*}
\E\left(\widehat{\mu}-\mu\mid X,\beta\right) & =\widetilde{L}_{w,\lambda}^{-1}\left(-\lambda_{w}\mu+\lambda_{f}BD_{f,\lambda}^{-1}\phi\right),\\
\E\left(\widehat{\phi}-\phi\mid X,\beta\right) & =\widetilde{L}_{f,\lambda}^{-1}\left(\lambda_{w}B^{\top}D_{w,\lambda}^{-1}\mu-\lambda_{f}\phi\right).
\end{align*}
The variance of $(\widehat{\mu},\widehat{\phi})$ follows with 
\begin{align*}
\V\left(\widehat{\mu}\mid X,\beta\right) & =\sigma^{2}\widetilde{L}_{w,\lambda}^{-1}-\sigma^{2}\widetilde{L}_{w,\lambda}^{-1}\left[\lambda_{w}I_{n}+\lambda_{f}BD_{f,\lambda}^{-2}B^{\top}\right]\widetilde{L}_{w,\lambda}^{-1},\\
\V\left(\widehat{\phi}\mid X,\beta\right)= & \sigma^{2}\widetilde{L}_{f,\lambda}^{-1}-\sigma^{2}\widetilde{L}_{f,\lambda}^{-1}\left[\lambda_{w}B^{\top}D_{w,\lambda}^{-2}B+\lambda_{f}I_{p}\right]\widetilde{L}_{f,\lambda}^{-1},\\
\Cov\left(\widehat{\mu},\widehat{\phi}\mid X,\beta\right) & =-\sigma^{2}\widetilde{L}_{w,\lambda}^{-1}BD_{f,\lambda}^{-1}+\sigma^{2}\widetilde{L}_{w,\lambda}^{-1}\left(\lambda_{w}D_{w,\lambda}^{-1}B+\lambda_{f}BD_{f,\lambda}^{-1}\right)\widetilde{L}_{f,\lambda}^{-1}.
\end{align*}

\section{A stochastic block model}

\label{sec:SBM}

To study the statistical properties of these estimators, we need a
model for the graph of worker-firm matches. We model
the bipartite graph as a Degree-Corrected Stochastic Block Model (\cite{DasguptaHopcroftMcSherry2004},
\cite{KarrerNewman2011}, \cite{LarremoreClausetJacobs2014},
\cite{RazaeeAminiLi2019}).

\paragraph*{Node types}

We first divide an arbitrary number of $n$ workers and $p$ firms
into $K$ groups. Let $k\in\{1,...,K\}$ index worker types and $\ell\in\{1,...,K\}$
index firm types. We assume the same number of groups for simplicity.
Workers initially draw a type $k_{i}=k$ independently with probability
$\pi_{w}(k)$, and firms draw a type $\ell_{j}=\ell$ with probability
$\pi_{f}(\ell)$. Let $n_{k}$ denote the number of workers of type
$k$.
Second, we introduce a parameter $\theta_{j}$ to control the expected
degree of firm $j$ with the restriction that the $\theta_{j}'s$
sum to one within firm groups: $\sum_{j'}\theta_{j'}\delta_{\ell_{j'}\ell}=1$.
Parameter $\theta_{j}$ can therefore be interpreted as the probability
of drawing firm $j$ in its group $\ell$ conditional on drawing any
firm from this group. For workers, we assume uniform sampling given
type (i.e. with probability $1/n_{k_{i}}$). We could introduce another
fixed effect $\theta$ to control worker sampling given type. By using
a degree correction for one node type and not for the other type,
we cover both the standard SBM and its degree-corrected extension.

\paragraph*{Network}

The block-structure of the network is governed by the affinity matrix
$C$, with $C(k,\ell)\ge0$. Independently for each worker-firm couple
$(i,j)$, we draw $d_{ij}$ edges between $i$ and $j$ from a Bernoulli
distribution where the probability of an edge is $p_{ij}=\frac{1}{n_{k_{i}}}\theta_{j}C(k_{i},\ell_{j})$. The expected number of links between workers of type $k$ and firms
of type $\ell$ is 
\[
\E\left(\sum_{i,j}d_{ij}\delta_{k_{i}k}\delta_{\ell_{j}\ell}\right)=\sum_{i,j}p_{ij}\delta_{k_{i}k}\delta_{\ell_{j}\ell}=C(k,\ell),
\]
as $\sum_{k}\delta_{k_{i}k}=n_{k_{i}}$ and $\sum_{j}\delta_{\ell_{j}\ell}\theta_{j}=1$.
The expected firm degree is 
\[
\E\left(d_{\cdot j}\right)=\E\left(\sum_{i}d_{ij}\right)=\sum_{i}p_{ij}=\sum_{i}\frac{1}{n_{k_{i}}}\theta_{j}C(k_{i},\ell_{j})=\theta_{j}C(\cdot,\ell_{j}),
\]
where $C(\cdot,\ell)=\sum_{k}C(k,\ell)$, and the expected worker
degree is 
\[
\E\left(d_{i\cdot}\right)=\E\left(\sum_{j}d_{ij}\right)=\sum_{j}p_{ij}=\frac{1}{n_{k_{i}}^{0}}\sum_{\ell}\sum_{j}\theta_{j}\delta_{\ell_{j}\ell}C(k_{i},\ell)=\frac{1}{n_{k_{i}}}C(k_{i},\cdot),
\]
where $C(k,\cdot)=\sum_{\ell}C(k,\ell)$. 

For example, in the simulations we use 
\[
C=c\frac{p}{K}\left[I_{K}+\delta(J_{K}-I_{K})\right],
\]
where $c$ is a constant, $J_{K}=\mathbf{1}_{K}\mathbf{1}_{K}^{\top}$
and $\delta$ is a tuning parameter: $\delta=0$ means perfect segregation
and $\delta=1$ means indifference. With this specification, the number
of connections per pair $(k,\ell)$ is proportional to the number
of firms. If workers are uniformly assigned to groups, then $n_{k}\simeq n/K$
and the average worker degree is proportional to $p/n=\gamma$. It
is therefore fixed and does not grow with $p$ and $n$ if the numbers
of workers and firms grow at the same speed. 

The realized graph typically features one main connected component containing orders of magnitude more observations than any other components. For identification reasons, it is usual to use it as estimation sample. 
Then, $n$ and $p$ denote the numbers of workers
and firms in the estimation sample, and $n_{k}$ is the number
of workers of type $k$ in the sample.

\section{Asymptotic theory for the network}

\label{sec:asymptotics_network}

We develop the asymptotic theory for a large network with many firms
and workers. Worker and firm types are given, as well as parameters
$\theta_{j}$ and $C$. The edge weights $d_{ij}=\text{Bernoulli}(p_{ij})$,
with $p_{ij}=\frac{1}{n_{k_{i}}}\theta_{j}C(k_{i},\ell_{j})$, are
the only stochastic variables. The following analysis does not assume
the selection of the biggest connected component. 

\subsection{Deterministic equivalent network}

Following \cite{QinRohe2013}, we use a calligraphic font to indicate
expectations with respect to $B=(d_{ij})$. We thus denote 
\begin{equation}
\mathfrak{B}:=\E B=(p_{ij})=\left[\tfrac{1}{n_{k_{i}}}\theta_{j}C(k_{i},\ell_{j})\right]=\text{diag}\left(\tfrac{1}{n_{k_{i}}}\right)Z_{w}CZ_{f}^{\top}\text{diag}\left(\theta_{j}\right),\label{eq:cal(B)}
\end{equation}
where $Z_{w}=\left(\delta_{k_{i}k}\right)\in\{0,1\}^{n\times K}$
and $Z_{f}=\left(\delta_{\ell_{j}\ell}\right)\in\{0,1\}^{p\times K}$
are the selection matrices containing the information on the community
memberships of workers and firms. Moreover, 
\begin{align*}
\mathfrak{B}^{\top}\mathfrak{B}=\left[ \theta_{j} \sum\nolimits_{k}\frac{C(k,\ell_{j})C(k,\ell_{j'})}{n_{k}} \theta_{j'} \right]=\text{diag}\left(\theta_{j}\right)Z_{f}C^{\top}\text{diag}\left(\tfrac{1}{n_{k}}\right)CZ_{f}^{\top}\text{diag}\left(\theta_{j}\right).
\end{align*}
Using the normalization $\sum_{j}\theta_{j}\delta_{\ell_{j}\ell}=1$
and the notation $C(k,\cdot)=\sum_{\ell}C(k,\ell)$, we have 
\begin{align*}
\mathfrak{D}_{w,\lambda} & =\text{diag}\left(\tfrac{1}{n_{k_{i}}}\sum\nolimits_{j}\theta_{j}C(k_{i},\ell_{j})+\lambda_{w}\right)=\text{diag}\left(\tfrac{1}{n_{k_{i}}}C(k_{i},\cdot)+\lambda_{w}\right).
\end{align*}
And using the notation $C(\cdot,\ell)=\sum_{k}C(k,\ell)$ and the
fact that 
\[
\sum_{i}\tfrac{1}{n_{k_{i}}}C(k_{i},\ell_{j})=\sum_{k}\sum_{i}\frac{\delta_{k_{i}k}}{n_{k_{i}}}C(k,\ell_{j})=\sum_{k}C(k,\ell_{j}),
\]
we have 
\begin{align*}
\mathfrak{D}_{f,\lambda} & =\text{diag}\left(\theta_{j}\sum\nolimits_{i}\tfrac{1}{n_{k_{i}}}C(k_{i},\ell_{j})+\lambda_{f}\right)=\text{diag}\left(\theta_{j}C(\cdot,\ell_{j})+\lambda_{f}\right).
\end{align*}

Let $\mathfrak{E}_{\lambda}=\mathfrak{D}_{w,\lambda}^{-1/2}\mathfrak{B}\mathfrak{D}_{f,\lambda}^{-1/2}.$
The generic element of $\mathfrak{E}_{\lambda}$ is 
\begin{align*}
\left(\mathfrak{E}_{\lambda}\right)_{ij} & =\tfrac{\theta_{j}C(k_{i},\ell_{j})/n_{k_{i}}}{\sqrt{C(k_{i},\cdot)/n_{k_{i}}+\lambda_{w}}\sqrt{\theta_{j}C(\cdot,\ell_{j})+\lambda_{f}}}\\
 & =\tfrac{1}{n_{k_{i}}}\sqrt{\tfrac{C(k_{i},\cdot)}{\tfrac{1}{n_{k_{i}}}C(k_{i},\cdot)+\lambda_{w}}}\tfrac{C(k_{i},\ell_{j})}{\sqrt{C(k_{i},\cdot)C(\cdot,\ell_{j})}}\theta_{j}\sqrt{\tfrac{C(\cdot,\ell_{j})}{\theta_{j}C(\cdot,\ell_{j})+\lambda_{f}}}=\omega_{\lambda,k_{i}}\widetilde{C}(k_{i},\ell_{j})\phi_{\lambda,j},
\end{align*}
denoting $\widetilde{C}(k,\ell)=\frac{C(k,\ell)}{\sqrt{C(k,\cdot)C(\cdot,\ell)}}$,
$\omega_{\lambda,k}=\frac{1}{n_{k}}\sqrt{\frac{C(k,\cdot)}{\frac{1}{n_{k_{i}}}C(k,\cdot)+\lambda_{w}}}$
and $\phi_{\lambda,j}=\theta_{j}\sqrt{\frac{C(\cdot,\ell_{j})}{\theta_{j}C(\cdot,\ell_{j})+\lambda_{f}}}$. 
The regularization reduces the noise impact of worker communities
with few employment links per capita ($\frac{1}{n_{k}}C(k,\cdot)$
small) and firm groups with few employees per firm ($\theta_{j}C(\cdot,\ell_{j})$
small).

Then, define the regularized adjacency and Laplacian matrices
\begin{align*}
\mathcal{A}_{f,\lambda} & :=\mathfrak{E}_{\lambda}^{\top}\mathfrak{E}_{\lambda}=\left[\phi_{\lambda,j}\left(\sum_{k}n_{k}\widetilde{C}(k,\ell_{j})\omega_{\lambda,k}^{2}\widetilde{C}(k,\ell_{j'})\right)\phi_{\lambda,j'}\right]_{j,j'},\quad\mathfrak{L}_{f,\lambda}=I_{p}-\mathcal{A}_{f,\lambda},\\
\mathcal{A}_{w,\lambda} & :=\mathfrak{E}_{\lambda}\mathfrak{E}_{\lambda}^{\top}=\left[\omega_{\lambda,i}\left(\sum_{j}\widetilde{C}(k_{i},\ell_{j})\phi_{\lambda,j}^{2}\widetilde{C}(k_{i'},\ell_{j})\right)\omega_{\lambda,i'}\right]_{i,i'},\quad\mathfrak{L}_{w,\lambda}=I_{p}-\mathcal{A}_{w,\lambda}.
\end{align*}

The following lemma shows that the matrix $\mathfrak{E}_{\lambda}$
contains the community structure of the random graph and that the
analog Laplacian $\mathfrak{L}_{f,\lambda}=I_{p}-\mathfrak{E}_{\lambda}^{\top}\mathfrak{E}_{\lambda}$
has its eigenvalues bounded away from 0. 
\begin{lemma}
\label{lem:max-eig-cal(E)'cal(E)} 1) $\mathfrak{E}_{\lambda}$, $\mathcal{A}_{f,\lambda}$
and $\mathcal{A}_{w,\lambda}$ have rank $K$. 2) If $\lambda_{f}>0$,
$\eigmin(\mathfrak{L}_{f,\lambda})\ge\frac{\lambda_{f}}{\max_{j}\theta_{j}C(\cdot,\ell_{j})+\lambda_{f}}>0$
and $\left\Vert \mathfrak{E}_{\lambda}\right\Vert \le\sqrt{\frac{\max_{j}\theta_{j}C(\cdot,\ell_{j})}{\max_{j}\theta_{j}C(\cdot,\ell_{j})+\lambda_{f}}}<1$.
3) If $\lambda_{w}>0$, $\eigmin(\mathfrak{L}_{w,\lambda})\ge\frac{\lambda_{w}}{\max_{k}\frac{1}{n_{k}}C(k,\cdot)+\lambda_{w}}>0$
and $\left\Vert \mathfrak{E}_{\lambda}\right\Vert \le\sqrt{\frac{\max_{k}\frac{1}{n_{k}}C(k,\cdot)}{\max_{k}\frac{1}{n_{k}}C(k,\cdot)+\lambda_{w}}}<1$. 

\end{lemma}

\subsection{Asymptotic bounds for the Laplacian }

The main tool used to derive asymptotic bounds is a version of Bernstein
inequality for random matrices (see \cite{ChungRadcliffe2011}
and \cite{Tropp2015}, Theorem 6.1.1). Let $S_{1},...,S_{n}$ be
independent, centered random matrices with common dimension $d_{1}\times d_{2}$
and assume that each one is uniformly bounded: 
\[
\E S_{k}=0\quad\text{and}\quad\left\Vert S_{k}\right\Vert \le L,\quad\forall k=1,...,n.
\]
Introduce the sum $Z=\sum_{k=1}^{n}S_{k}$ and let $v(Z)$ denote
the matrix variance statistic of the sum: 
\[
v(Z)=\max\left\{ \left\Vert \E\left(ZZ^{\top}\right)\right\Vert ,\left\Vert \E\left(Z^{\top}Z\right)\right\Vert \right\} =\max\left\{ \left\Vert \sum_{k=1}^{n}\E\left(S_{k}S_{k}^{\top}\right)\right\Vert ,\left\Vert \sum_{k=1}^{n}\E\left(S_{k}^{\top}S_{k}\right)\right\Vert \right\} .
\]
Then, 
\[
\Pr\left\{ \left\Vert Z\right\Vert \ge t\right\} \le(d_{1}+d_{2})\exp\left(-{\frac{1}{2}t^{2}}/{\left(v(Z)+\frac{1}{3}Lt\right)}\right),\quad\forall t\ge0.
\]
Furthermore, 
\[
\E Z\le\sqrt{2v(Z)\ln(d_{1}+d_{2})}+\frac{1}{3}L\ln(d_{1}+d_{2}).
\]

We first prove the following concentration inequality that adapts
Theorem 4.1 of \cite{QinRohe2013} to regularized Laplacians of bipartite
graphs. The ridge coefficients $\lambda_{w}$ and $\lambda_{f}$ allow
us to choose $M_{w}\vee M_{f}:=\max(M_{w},M_{f})$ large enough to
maintain the wedge $\left\Vert E_{\lambda}-\mathfrak{E}_{\lambda}\right\Vert $
less than any chosen value with any specified probability. 
\begin{theorem}
\label{thm:ConcentrationLaplacian}Let $M_{w}=\left(\underline{\delta}_{w}+\lambda_{w}\right)^{-1}$
and $M_{f}=\left(\underline{\delta}_{f}+\lambda_{f}\right)^{-1}$,
where $\underline{\delta}_{w}=\min_{k}\frac{1}{n_{k}}C(k,\cdot)$
and $\underline{\delta}_{f}=\min_{j}\theta_{j}C(\cdot,\ell_{j})$
are the minimum expected worker and firm degrees. For any $\epsilon>0$,
if $M=M_{f}\vee M_{w}\le \left(3\ln\frac{n+p}{\epsilon} \right)^{-1}$, then
with probability at least $1-\frac{3+4\gamma}{1+\gamma}\epsilon$
(where $\gamma=p/n$), 
\begin{gather*}
\left\Vert E_{\lambda}-\mathfrak{E}_{\lambda}\right\Vert \le4t,\quad\left\Vert L_{w,\lambda}-\mathfrak{L}_{w,\lambda}\right\Vert \le8t,\quad\left\Vert L_{f,\lambda}-\mathfrak{L}_{f,\lambda}\right\Vert \le8t,
\end{gather*}
where $t=\sqrt{3M\ln\frac{n+p}{\epsilon}}\le1$. 
\end{theorem}
\begin{remark}
\label{Rk: rates} For example, let $\epsilon=(n+p)^{-\nu}$ with
$\nu\in(0,1)$. It then suffices that $\lambda_{w}\wedge\lambda_{f}:=\min(\lambda_{w},\lambda_{f})\ge3(1+\nu)\ln(n+p)$
for the theorem's conditions to be satisfied, even if the minimum
expected worker degree $\underline{\delta}_{w}=\min_{k}\frac{1}{n_{k}}C(k,\cdot)$
and the minimum expected firm degree $\underline{\delta}_{f}=\min_{j}\theta_{j}C(\cdot,\ell_{j})$
are much lower.\footnote{\cite{ChungRadcliffe2011} and \cite{ChaudhuriGrahamTsiatas2012}
assume that the minimum expected degree is $\delta\ge\nu\ln n$, where
$n$ is the number of nodes of the graph. As highlighted by \cite{QinRohe2013},
the regularization allows to remove the need for increasing node degrees.} In particular, choosing $\lambda_{w}\wedge\lambda_{f}\propto\left(\ln(n+p)\right)^{1+\nu'}$
with $\nu'>0$ implies that the Laplacians converge in probability
towards their expected values for the spectral norm (as $t\to0$ when
$n+p\to\infty$). 
\end{remark}
The next theorem shows that the previous property passes to the inverse. 
\begin{theorem}
\label{thm:ConcentrationInverseLaplacian}For any $\epsilon>0$, under
the same conditions as in Theorem \ref{thm:ConcentrationLaplacian},
with probability at least $1-\frac{3+5\gamma}{1+\gamma}\epsilon$,
\[
\left\Vert L_{f,\lambda}^{-1}-\mathfrak{L}_{f,\lambda}^{-1}\right\Vert \le16t\left(\frac{\overline{\delta}_{f}+\lambda_{f}}{\lambda_{f}}\right)^{2},
\]
and with probability at least $1-4\epsilon$, 
\[
\left\Vert L_{w,\lambda}^{-1}-\mathfrak{L}_{w,\lambda}^{-1}\right\Vert \le16t\left(\frac{\overline{\delta}_{w}+\lambda_{w}}{\lambda_{w}}\right)^{2},
\]
where $\overline{\delta}_{w}=\max_{k}\frac{1}{n_{k}}C(k,\cdot)$ and
$\overline{\delta}_{f}=\max_{j}\theta_{j}C(\cdot,\ell_{j})$ denote
the maximum expected worker and firm degrees. 
\end{theorem}
\begin{remark}
We obtain these bounds using the following inequality: 
\[
\left\Vert L_{f,\lambda}^{-1}-\mathfrak{L}_{f,\lambda}^{-1}\right\Vert =\left\Vert L_{f,\lambda}^{-1}\left(L_{f,\lambda}-\mathfrak{L}_{f,\lambda}\right)\mathfrak{L}_{f,\lambda}^{-1}\right\Vert \le\left\Vert L_{f,\lambda}^{-1}\right\Vert \left\Vert \mathfrak{L}_{f,\lambda}^{-1}\right\Vert \left\Vert L_{f,\lambda}-\mathfrak{L}_{f,\lambda}\right\Vert .
\]
Then, we use Lemma \ref{lem:max-eig-cal(E)'cal(E)} to bound $\eigmin(\mathfrak{L}_{f,\lambda})\ge\frac{\lambda_{f}}{\overline{\delta}_{f}+\lambda_{f}}$,
and therefore $\eigmax(\mathfrak{L}_{f,\lambda}^{-1})\ge\frac{\overline{\delta}_{f}+\lambda_{f}}{\lambda_{f}}$,
and similarly for the empirical Laplacian. Without the regularization,
the smallest eigenvalue of the Laplacian is equal to zero. It is easy
to avoid this problem by using a Moore-Penrose inverse. However, it
still remains that the second lowest eigenvalue is not necessarily bounded away from zero, even if the graph
is connected. A small second
lowest eigenvalue of the Laplacian reflects the presence of bottlenecks
in the network, and bottlenecks get more frequent with sparsity (\cite{ZhangRohe2018}). 
\end{remark}
\begin{remark}
In many empirical setups, degrees will not grow with the number of
nodes. In the economic wage application, worker and firm degrees grow
with the number of observation periods, which is a fixed, rather small
number. It follows that the concentration bounds are not greatly degraded
by the inverse operation when the number of nodes increases. 
\end{remark}
We end this section by showing a similar concentration inequality
for the inverse of the un-normalized Laplacian matrices, 
\begin{align*}
\widetilde{L}_{f,\lambda}^{-1} & =D_{f,\lambda}^{-1/2}L_{f,\lambda}^{-1}D_{f,\lambda}^{-1/2}=\left(D_{f,\lambda}-B^{\top}D_{w,\lambda}^{-1}B\right)^{-1},\\
\widetilde{L}_{w,\lambda}^{-1} & =D_{w,\lambda}^{-1/2}L_{w,\lambda}^{-1}D_{w,\lambda}^{-1/2}=\left(D_{w,\lambda}-BD_{f,\lambda}^{-1}B^{\top}\right)^{-1}.
\end{align*}

\begin{theorem}[Concentration of the inverse of the regularized un-normalized Laplacians]
\label{thm:ConcentrationInvUnnormalizedLap}For any $\epsilon>0$,
under the same conditions as in Theorem \ref{thm:ConcentrationLaplacian},
with probability at least $1-\frac{3+9\gamma}{1+\gamma}\epsilon$,
\[
\left\Vert \widetilde{L}_{f,\lambda}^{-1}-\widetilde{\mathfrak{L}}_{f,\lambda}^{-1}\right\Vert =\left\Vert D_{f,\lambda}^{-1/2}L_{f,\lambda}^{-1}D_{f,\lambda}^{-1/2}-\mathfrak{D}_{f,\lambda}^{-1/2}\mathfrak{L}_{f,\lambda}^{-1}\mathfrak{D}_{f,\lambda}^{-1/2}\right\Vert \le\left(5+16\frac{\overline{\delta}_{f}+\lambda_{f}}{\underline{\delta}_{f}+\lambda_{f}}\right)\frac{\overline{\delta}_{f}+\lambda_{f}}{\lambda_{f}^{2}}t,
\]
and, with probability at least $1-\frac{8+4\gamma}{1+\gamma}\epsilon$,
\[
\left\Vert \widetilde{L}_{w,\lambda}^{-1}-\widetilde{\mathfrak{L}}_{w,\lambda}^{-1}\right\Vert =\left\Vert D_{w,\lambda}^{-1/2}L_{w,\lambda}^{-1}D_{w,\lambda}^{-1/2}-\mathfrak{D}_{w,\lambda}^{-1/2}\mathfrak{L}_{w,\lambda}^{-1}\mathfrak{D}_{w,\lambda}^{-1/2}\right\Vert \le\left(5+16\frac{\overline{\delta}_{w}+\lambda_{w}}{\underline{\delta}_{w}+\lambda_{w}}\right)\frac{\overline{\delta}_{w}+\lambda_{w}}{\lambda_{w}^{2}}t.
\]
\end{theorem}
\begin{remark}
If the expected node degrees do not increase with the number of nodes,
but the regularization parameters grow with the network size, then
the bounds are of the order of $t/\lambda_{f}$ and $t/\lambda_{w}$
and thus tighter for the inverse of the un-normalized Laplacian.
This makes sense as the un-normalized Laplacian is the normalized
Laplacian multiplied by the regularized degree matrix, whose minimum
eigenvalue is bounded from below by one over the regularization parameter. 
\end{remark}

\section{The ridge regression}

\label{sec:asymptotics_ridge}

This section shows that the concentration inequalities of the preceding section guarantee that the ridge fixed effects converge to well-defined limits, both in terms of bias and variance.

Rewrite the ridge estimator as 
\begin{align*}
\widehat{\mu}-\mu & =\widetilde{L}_{w,\lambda}^{-1}\left[-\lambda_{w}\mu+\lambda_{f}BD_{f,\lambda}^{-1}\phi+\left(W^{\top}-BD_{f,\lambda}^{-1}F^{\top}\right)U\right],\\
\widehat{\phi}-\phi & =\widetilde{L}_{f,\lambda}^{-1}\left[\lambda_{w}B^{\top}D_{w,\lambda}^{-1}\mu-\lambda_{f}\phi+\left(F^{\top}-B^{\top}D_{w,\lambda}^{-1}W^{\top}\right)U\right],
\end{align*}
using $Y=W\mu+F\phi+U$, where $U$ has independent entries with mean
0 and variance $\sigma^{2}$.

We assume that the parameters $(\mu,\phi)$ are random
with 
\[
\mu=Z_{w}\mu^{*}+U_{w},\quad\phi=Z_{f}\phi^{*}+U_{f},
\]
where $Z_{w}$ and $Z_{f}$ are the $n\times K$ and $p\times K$
matrices indicating worker and firm communities, $\mu^{*}$ and $\phi^{*}$
are $K$-vectors of group fixed effects, and $U_{w}$ and $U_{f}$
are vectors of independent components with mean zero and variance
$\sigma_{w}^{2}$ and $\sigma_{f}^{2}$. It is further assumed that
$U_{w}$ and $U_{f}$ are independent of $U$.

\subsection{Bias}

Given the community structure $Z=(Z_{w},Z_{f})$ and $\beta^{*}=(\mu^{*},\phi^{*})$,
and given the network structure $X=(W,F)$, the biases on $\widehat{\mu}$
and $\widehat{\phi}$ are 
\begin{align*}
b_{\mu,\lambda}  &:=\E\left(\widehat{\mu}-\mu\mid X,Z,\beta^{*}\right)=\widetilde{L}_{w,\lambda}^{-1}\left(-\lambda_{w}Z_{w}\mu^{*}+\lambda_{f}BD_{f,\lambda}^{-1}Z_{f}\phi^{*}\right),\label{eq:bias}\\
b_{\phi,\lambda}  &:=\E\left(\widehat{\phi}-\phi \mid X,Z,\beta^{*}\right)=\widetilde{L}_{f,\lambda}^{-1}\left(-\lambda_{f}Z_{f}\phi^{*}+\lambda_{w}B^{\top}D_{w,\lambda}^{-1}Z_{w}\mu^{*}\right).
\end{align*}
Define the deterministic bias limits on $\mu$ and $\phi$ as the
ones obtained using the expected adjacency matrix $\mathfrak{B}$,
\begin{align*}
\mathfrak{b}_{\mu,\lambda} & :=\widetilde{\mathfrak{L}}_{w,\lambda}^{-1}\left(-\lambda_{w}Z_{w}\mu^{*}+\lambda_{f}\mathfrak{B}\mathfrak{D}_{f,\lambda}^{-1}Z_{f}\phi^{*}\right)\\
b_{\phi,\lambda} & :=\widetilde{L}_{f,\lambda}^{-1}\left(-\lambda_{f}Z_{f}\phi^{*}+\lambda_{w}B^{\top}D_{w,\lambda}^{-1}Z_{w}\mu^{*}\right).
\end{align*}
The following theorem proves that $\mathfrak{b}_{\mu,\lambda}$ and
$\mathfrak{b}_{\phi,\lambda}$ are asymptotically good predictors
of the bias $b_{\mu,\lambda}$ and $b_{\phi,\lambda}$ in MSE
respectively, as long as $\lambda_{w},\lambda_{f}$ grow faster than
$\ln(n+p)$. 
\begin{theorem}
\label{thm:ConcentrationRidgeBias}For any $\epsilon>0$, under the
same conditions as in Theorem \ref{thm:ConcentrationLaplacian}, with
probability at least $1-\frac{11+7\gamma}{1+\gamma}\epsilon$, 
\begin{multline*}
\left\Vert b_{\mu,\lambda}-\mathfrak{b}_{\mu,\lambda}\right\Vert \le\left(5+16\frac{\overline{\delta}_{w}+\lambda_{w}}{\underline{\delta}_{w}+\lambda_{w}}\right)\frac{\overline{\delta}_{w}+\lambda_{w}}{\lambda_{w}^{2}}t\left(\lambda_{w}\sqrt{n}\left\Vert \mu^{*}\right\Vert +\lambda_{f}\sqrt{\frac{\overline{\delta}_{w}+\lambda_{w}}{\underline{\delta}_{f}+\lambda_{f}}}\sqrt{p}\left\Vert \phi^{*}\right\Vert \right)\\
+2\lambda_{f}\frac{\overline{\delta}_{w}+\lambda_{w}}{\lambda_{w}^{2}}\left(\sqrt{\chi_{f}}+2\sqrt{\frac{\overline{\delta}_{w}+\lambda_{w}}{\lambda_{f}}}\right)t\sqrt{p}\left\Vert \phi^{*}\right\Vert ,
\end{multline*}
where $\chi_{f}=\max\left(1,\frac{\overline{\delta}_{w}}{\underline{\delta}_{f}+\lambda_{f}}\right)$.
And with probability at least $1-\frac{6+12\gamma}{1+\gamma}\epsilon$,
\begin{multline*}
\left\Vert b_{\phi,\lambda}-\mathfrak{b}_{\phi,\lambda}\right\Vert \le\left(5+16\frac{\overline{\delta}_{f}+\lambda_{f}}{\underline{\delta}_{f}+\lambda_{f}}\right)\frac{\overline{\delta}_{f}+\lambda_{f}}{\lambda_{f}^{2}}t\left(\lambda_{f}\sqrt{p}\left\Vert \phi^{*}\right\Vert +\lambda_{w}\sqrt{\frac{\overline{\delta}_{f}+\lambda_{f}}{\underline{\delta}_{w}+\lambda_{w}}}\sqrt{n}\left\Vert \mu^{*}\right\Vert \right)\\
+2\lambda_{w}\frac{\overline{\delta}_{f}+\lambda_{f}}{\lambda_{f}^{2}}\left(\sqrt{\chi_{w}}+2\sqrt{\frac{\overline{\delta}_{f}+\lambda_{f}}{\lambda_{w}}}\right)t\sqrt{n}\left\Vert \mu^{*}\right\Vert ,
\end{multline*}
where $\chi_{w}=\max\left(1,\frac{\overline{\delta}_{f}}{\underline{\delta}_{w}+\lambda_{w}}\right)$. 
\end{theorem}

\begin{remark}
If $p/n\rightarrow\gamma$ and if ridge parameters grow at the same
rate, faster than $\ln(n+p)$, say $\left(\ln(n+p)\right)^{1+\nu'}$
with $\nu'>0$; see Remark \ref{Rk: rates}. The bias bound is therefore
of order $\sqrt{n}t$. Hence, the root mean square error between $b_{\mu,\lambda}$
and $\mathfrak{b}_{\mu,\lambda}$ is bounded by a multiple of of $t$,
as in Theorem \ref{thm:ConcentrationLaplacian}, and therefore goes
to 0 in probability. 
\end{remark}

\subsection{Variance}

The  conditional variance of $\widehat{\mu}-\mu$ is
\begin{align*}
V_{w,\lambda}&=\V\left(\widehat{\mu}-\mu\mid X,Z,\beta^{*}\right)\\
&=\widetilde{L}_{w,\lambda}^{-1}\V\left(-\lambda_{w}U_{w}+\lambda_{f}BD_{f,\lambda}^{-1}U_{f}+W^{\top}U-BD_{f,\lambda}^{-1}F^{\top}U\mid X,Z,\beta^{*}\right)\widetilde{L}_{w,\lambda}^{-1}\\
&=\widetilde{L}_{w,\lambda}^{-1}\left[\left(\lambda_{w}^{2}\sigma_{w}^{2}-\lambda_{w}\sigma^{2}\right)I_{n}+\left(\lambda_{f}^{2}\sigma_{f}^{2}-\lambda_{f}\sigma^{2}\right)BD_{f,\lambda}^{-2}B^{\top}+\sigma^{2}\widetilde{L}_{w,\lambda}\right]\widetilde{L}_{w,\lambda}^{-1}\\
&=\sigma^{2}\widetilde{L}_{w,\lambda}^{-1}+\widetilde{L}_{w,\lambda}^{-1}\left[\left(\lambda_{w}^{2}\sigma_{w}^{2}-\lambda_{w}\sigma^{2}\right)I_{n}+\left(\lambda_{f}^{2}\sigma_{f}^{2}-\lambda_{f}\sigma^{2}\right)BD_{f,\lambda}^{-2}B^{\top}\right]\widetilde{L}_{w,\lambda}^{-1}.
\end{align*}
Define the deterministic equivalent of the variance matrix
as 
\[
\mathfrak{V}_{w,\lambda}=\sigma^{2}\widetilde{\mathfrak{L}}_{w,\lambda}^{-1}+\widetilde{\mathfrak{L}}_{w,\lambda}^{-1}\left[\left(\lambda_{w}^{2}\sigma_{w}^{2}-\lambda_{w}\sigma^{2}\right)I_{n}+\left(\lambda_{f}^{2}\sigma_{f}^{2}-\lambda_{f}\sigma^{2}\right)\mathfrak{B}\mathfrak{D}_{f,\lambda}^{-2}\mathfrak{B}^{\top}\right]\widetilde{\mathfrak{L}}_{w,\lambda}^{-1}.
\]
The firm-side analog of these variances are 
\begin{align*}
V_{f,\lambda} & =\sigma^{2}\widetilde{L}_{f,\lambda}^{-1}+\widetilde{L}_{f,\lambda}^{-1}\left[\left(\lambda_{f}^{2}\sigma_{f}^{2}-\lambda_{f}\sigma^{2}\right)I_{p}+\left(\lambda_{w}^{2}\sigma_{w}^{2}-\lambda_{w}\sigma^{2}\right)B^{\top}D_{w,\lambda}^{-2}B\right]\widetilde{L}_{f,\lambda}^{-1},\\
\mathfrak{V}_{f,\lambda} & =\sigma^{2}\widetilde{\mathfrak{L}}_{f,\lambda}^{-1}+\widetilde{\mathfrak{L}}_{f,\lambda}^{-1}\left[\left(\lambda_{f}^{2}\sigma_{f}^{2}-\lambda_{f}\sigma^{2}\right)I_{p}+\left(\lambda_{w}^{2}\sigma_{w}^{2}-\lambda_{w}\sigma^{2}\right)\mathfrak{B^{\top}}\mathfrak{D}_{w,\lambda}^{-2}\mathfrak{B}\right]\widetilde{\mathfrak{L}}_{f,\lambda}^{-1}.
\end{align*}
The next theorem shows that $\mathfrak{V}_{w,\lambda}$ and $\mathfrak{V}_{f,\lambda}$
are asymptotically equivalent to the variances $V_{w,\lambda}$ and
$V_{f,\lambda}$ as long as the ridge parameters grow at the same
rate, faster than $\ln(n+p)$.
\begin{theorem}
\label{thm:ConcentrationRidgeVariance}For any $\epsilon>0$, under
the same conditions as in Theorem \ref{thm:ConcentrationLaplacian},
with probability at least $1-\frac{13+6\gamma}{1+\gamma}\epsilon$,
\begin{multline*}
\left\Vert V_{w,\lambda}-\mathfrak{V}_{w,\lambda}\right\Vert \le\sigma^{2}\left(5+16\frac{\overline{\delta}_{w}+\lambda_{w}}{\underline{\delta}_{w}+\lambda_{w}}\right)\frac{\overline{\delta}_{w}+\lambda_{w}}{\lambda_{w}^{2}}t\\
+\left|\lambda_{w}^{2}\sigma_{w}^{2}-\lambda_{w}\sigma^{2}\right|\left(\frac{2}{\lambda_{w}}+\frac{1}{\underline{\delta}_{w}+\lambda_{w}}\right)\left(5+16\frac{\overline{\delta}_{w}+\lambda_{w}}{\underline{\delta}_{w}+\lambda_{w}}\right)\frac{\left(\overline{\delta}_{w}+\lambda_{w}\right)^{2}}{\lambda_{w}^{3}}t\\
+\left|\lambda_{f}^{2}\sigma_{f}^{2}-\lambda_{f}\sigma^{2}\right|\left(4\frac{1}{\lambda_{w}^ {}}\sqrt{\frac{1}{\lambda_{f}}}+\frac{1}{\underline{\delta}_{w}+\lambda_{w}}\sqrt{\frac{1}{\underline{\delta}_{f}+\lambda_{f}}}\right)\frac{\left(\overline{\delta}_{w}+\lambda_{w}\right)^{5/2}}{\lambda_{w}^{3}}\\
\times\left[\left(5+16\frac{\overline{\delta}_{w}+\lambda_{w}}{\underline{\delta}_{w}+\lambda_{w}}\right)\sqrt{\frac{\overline{\delta}_{w}+\lambda_{w}}{\underline{\delta}_{f}+\lambda_{f}}}+2\left(\sqrt{\chi_{w}}+2\sqrt{\frac{\overline{\delta}_{f}+\lambda_{f}}{\lambda_{w}}}\right)\right]t.
\end{multline*}
And similarly, with probability at least $1-\frac{5+13\gamma}{1+\gamma}\epsilon,$
\begin{multline*}
\left\Vert V_{f,\lambda}-\mathfrak{V}_{f,\lambda}\right\Vert \le\sigma^{2}\left(5+16\frac{\overline{\delta}_{f}+\lambda_{f}}{\underline{\delta}_{f}+\lambda_{f}}\right)\frac{\overline{\delta}_{f}+\lambda_{f}}{\lambda_{f}^{2}}t\\
+\left|\lambda_{f}^{2}\sigma_{f}^{2}-\lambda_{f}\sigma^{2}\right|\left(\frac{2}{\lambda_{f}}+\frac{1}{\underline{\delta}_{f}+\lambda_{f}}\right)\left(5+16\frac{\overline{\delta}_{f}+\lambda_{f}}{\underline{\delta}_{f}+\lambda_{f}}\right)\frac{\left(\overline{\delta}_{f}+\lambda_{f}\right)^{2}}{\lambda_{f}^{3}}t\\
+\left|\lambda_{w}^{2}\sigma_{w}^{2}-\lambda_{w}\sigma^{2}\right|\left(4\frac{1}{\lambda_{f}}\sqrt{\frac{1}{\lambda_{w}}}+\frac{1}{\underline{\delta}_{f}+\lambda_{f}}\sqrt{\frac{1}{\underline{\delta}_{w}+\lambda_{w}}}\right)\frac{\left(\overline{\delta}_{f}+\lambda_{f}\right)^{5/2}}{\lambda_{f}^{3}}\\
\times\left[\left(5+16\frac{\overline{\delta}_{f}+\lambda_{f}}{\underline{\delta}_{f}+\lambda_{f}}\right)\sqrt{\frac{\overline{\delta}_{f}+\lambda_{f}}{\underline{\delta}_{w}+\lambda_{w}}}+2\left(\sqrt{\chi_{w}}+2\sqrt{\frac{\overline{\delta}_{f}+\lambda_{f}}{\lambda_{w}}}\right)\right]t.
\end{multline*}
\end{theorem}
\begin{remark}
If $p/n\rightarrow\gamma$ and if ridge parameters grow at the same
rate faster than $\ln(n+p)$, say $\left(\ln(n+p)\right)^{1+\nu'}$
with $\nu'>0$, the first term on the right hand side of the inequalities
is of order $t/\lambda_{w}$ or $t/\lambda_{f}$, and the last two
terms are of order $t$. 
\end{remark}

\subsection{Prediction}

Suppose that we draw another network and outcomes from the same DGP,
say $\widetilde{X}=(\widetilde{W},\widetilde{F})$ and $\widetilde{Y}=\widetilde{W}\mu+\widetilde{F}\phi+\widetilde{U}$.
The prediction $SSE$ is 
\begin{multline*}
SSE=\E\left[\left(\widetilde{W}(\mu-\widehat{\mu})+\widetilde{F}(\phi-\widehat{\phi})+\widetilde{U}\right)^{\top}\left(\widetilde{W}(\mu-\widehat{\mu})+\widetilde{F}(\phi-\widehat{\phi})+\widetilde{U}\right)\mid X,Z,\beta^{*}\right] \\
=\sigma^{2}+
\E\left[(\mu-\widehat{\mu})^{\top}\mathfrak{D}_{w}(\mu-\widehat{\mu})+2(\mu-\widehat{\mu})^{\top}\mathfrak{B}(\phi-\widehat{\phi})+(\phi-\widehat{\phi})^{\top}\mathfrak{D}_{f}(\phi-\widehat{\phi})\mid X,Z,\beta^{*}\right],
\end{multline*}
as the expectations of $\widetilde{W}^{\top}\widetilde{W}$, $\widetilde{F}^{\top}\widetilde{F}$
and $\widetilde{W}^{\top}\widetilde{F}$ are the same as before.

The first non trivial term is 
\begin{multline*}
\E\left[(\mu-\widehat{\mu})^{\top}\mathfrak{D}_{w}(\mu-\widehat{\mu})\mid X,Z,\beta^{*}\right]=\E\left(\widehat{\mu}-\mu\mid X,Z,\beta^{*}\right)^{\top}\mathfrak{D}_{w}\E\left(\widehat{\mu}-\mu\mid X,Z,\beta^{*}\right)\\
+\text{tr}\left[\mathfrak{D}_{w}\V\left(\widehat{\mu}-\mu|X,Z,\beta^{*}\right)\right]=b_{\mu,\lambda}^{\top}\mathfrak{D}_{w}b_{\mu,\lambda}+\text{tr}\left[\mathfrak{D}_{w}V_{w,\lambda}\right].
\end{multline*}
The other terms follow similarly, yielding 
\[
SSE=b_{\mu,\lambda}^{\top}\mathfrak{D}_{w}b_{\mu,\lambda}+\text{tr}\left[\mathfrak{D}_{w}V_{w,\lambda}\right]+2b_{\mu,\lambda}^{\top}\mathfrak{B}b_{\phi,\lambda}+\text{tr}\left[\mathfrak{B}C_{\lambda}\right]+b_{\phi,\lambda}^{\top}\mathfrak{D}_{f}b_{\phi,\lambda}+\text{tr}\left[\mathfrak{D}_{f}V_{f,\lambda}\right],
\]
where $C_{\lambda}=\E\left[\left(\widehat{\mu}-\mu-b_{\mu,\lambda}\right)\left(\widehat{\phi}-\phi-b_{\phi,\lambda}\right)^{\top}\mid X,Z,\beta^{*}\right]$
denotes the covariance matrix. A deterministic equivalent is obtained
by replacing biases and variances by their deterministic equivalents.

Consider the first term: 
\[
\left|b_{\mu,\lambda}^{\top}\mathfrak{D}_{w}b_{\mu,\lambda}-\mathfrak{b}_{\mu,\lambda}^{\top}\mathfrak{D}_{w}\mathfrak{b}_{\mu,\lambda}\right|\le\left\Vert \mathfrak{D}_{w}\right\Vert \left(\left\Vert b_{\mu,\lambda}\right\Vert +\left\Vert \mathfrak{b}_{\mu,\lambda}\right\Vert \right)\left\Vert b_{\mu,\lambda}-\mathfrak{b}_{\mu,\lambda}\right\Vert ,
\]
where $\left\Vert \mathfrak{D}_{w}\right\Vert =\bar{\delta}_{w}$
is the maximum worker expected degree. It is easy to show that under
the condition of the two previous theorems, if the ridge parameters
go to infinity at the same rate faster that $\ln(n+p)$, this quantity
is bounded of order $nt$.

Consider the second term: 
\[
\left|\text{tr}\left[\mathfrak{D}_{w}V_{w,\lambda}\right]-\text{tr}\left[\mathfrak{D}_{w}\mathfrak{V}_{w,\lambda}\right]\right|=\left|\text{tr}\left[\mathfrak{D}_{w}\left(V_{w,\lambda}-\mathfrak{V}_{w,\lambda}\right)\right]\right|\le\delta_{w}n\left\Vert V_{w,\lambda}-\mathfrak{V}_{w,\lambda}\right\Vert ,
\]
which is of the same oder as the first term.

And so on for all terms. This shows that the Mean Square Error ($SSE$
divided by number of observations $N$) is close to $n$ times the
average worker degree, and will converge to a deterministic equivalent
limit. Studying this limit and deriving optimal ridge parameters for
high dimensional datasets is outside the scope of this paper.

\section{Applications}

\subsection{Simulation results}\label{app:ComputeSimulate}

We start with an initial number of $n_0$ workers and $p_0$ firms. We simulate the model with 
\[
C=c\frac{p_0}{K}\left[I_{K}+\delta(J_{K}-I_{K})\right],
\]
where $c$ is a constant, $J_{K}=\mathbf{1}_{K}\mathbf{1}_{K}^{\top}$
and $\delta$ is a tuning parameter: $\delta=0$ means perfect segregation
and $\delta=1$ means indifference. 
The initial numbers of nodes are set such that $n_{0}=3p_{0}$ and the number
of communities is $K=5$. The community strength is controlled by
$\delta=0.1$. 
We draw $\theta_{j}$ independently from a Pareto distribution with
scale parameter $\alpha=2$ and minimum value $\theta_{min}=1$: $\Pr\{\theta_{j}>x\}=(\theta_{min}/x)^{\alpha}$.
Then, separately for each group $\ell$, we normalize $\theta_{j}$
for all firms $j$ belonging to group $\ell$ by dividing $\theta_{j}$
by $\sum_{j'}\theta_{j'}\delta_{\ell_{j'}\ell}$.  We select the biggest connected component and $n$ and $p$ denote 
the numbers of workers and firms that remain.

Lastly, we consider the following choices of wage parameters:
\begin{gather*}
\mu_{i}=k_{i}-1+\sqrt{2}\mathcal{N}(0,1),\quad
\phi_{j}=0.4(\ell_{j}-1)+\mathcal{N}(0,1).
\end{gather*}
This implies a form of positive sorting on wages (wage fixed effects
are positively correlated). The residual variance is $\sigma=2$.

\begin{table}
\caption{Characteristics of the simulated networks}

\begin{centering}
\begin{tabular}{cccc}
\hline 
 & $c$  & $1$  & 2 \tabularnewline
\hline 
Initial \#firms  & $p_{0}$  & 30000  & 4500 \tabularnewline
\#conncomp  &  & 11757  & 714 \tabularnewline
\#firms  & $p$  & 3059  & 3081 \tabularnewline
\#workers  & $n$  & 7249  & 6928 \tabularnewline
\#links  & $N=d_{\cdot\cdot}$  & 10341  & 11267 \tabularnewline
sparsity  & $\frac{N}{pn}$  & 0.05\%  & 0.05\% \tabularnewline
avg wkr degree  & $\frac{N}{n}$  & 1.4265  & 1.6263 \tabularnewline
avg firm degree  & $\frac{N}{p}$  & 3.3805  & 3.6569 \tabularnewline
\hline 
\end{tabular}
\par\end{centering}
\label{tab:sim_network} 
\end{table}

Table \ref{tab:sim_network} displays the main characteristics of
the simulated networks, one first time with $c=1$, a second time
with $c=2$. When $c=1$, the network is very sparse with many disconnected
components (after removing nodes with no connection). As $c$ increases,
the network becomes denser. We therefore use a much greater initial
value of firm nodes $p_{0}$ for $c=1$ than for $c=2$ in order to
generate a connected graph of roughly the same size $N$ close to
$10000$ observations.

\begin{table}
\setlength{\tabcolsep}{5pt}\caption{Simulated regression results}

\begin{centering}
\begin{tabular}{lcccccccc}
\hline 
 & \multicolumn{4}{c}{share of wage variance} & fixed effect  & outsample  & $\frac{\lambda_{w}}{N/n}$  & $\frac{\lambda_{f}}{N/p}$\tabularnewline
 & worker  & firm  & 2{*}cov  & residual  & correlation  & MSE  &  & \tabularnewline
\hline 
\multicolumn{7}{c}{$c=1$} &  & \tabularnewline
true  & 0.4259  & 0.1378  & 0.1212  & 0.3113  & 0.2500  &  &  & \tabularnewline
OLS  & 3.1724  & 2.8607  & -5.0343  & 0.0012  & -0.8356  & 7.4778  &  & \tabularnewline
OLS debiased  & -0.0453  & -0.2110  & 0.8959  & 0.3605  &  &  &  & \tabularnewline
ridge  & 0.2552  & 0.1028  & 0.1101  & 0.5319  & 0.3399  & 2.1289  & 0.58  & 0.83\tabularnewline
 &  &  &  &  &  &  &  & \tabularnewline
\multicolumn{7}{c}{$c=2$} &  & \tabularnewline
true  & 0.4363  & 0.1463  & 0.1062  & 0.3292  & 0.2101  &  &  & \tabularnewline
OLS  & 0.9798  & 0.5792  & -0.5949  & 0.0358  & -0.3948  & 3.6589  &  & \tabularnewline
OLS debiased  & 0.4089  & 0.1178  & 0.1526  & 0.3208  &  &  &  & \tabularnewline
ridge  & 0.2858  & 0.1026  & 0.1115  & 0.5002  & 0.3255  & 2.1167  & 0.48  & 0.53\tabularnewline
\hline 
\end{tabular}
\par\end{centering}
\label{tab:sim_regression} 
\end{table}

We show in Table \ref{tab:sim_regression} the corresponding simulations
for the log-wage variance decomposition. The first
row (``true'') shows the true variance decomposition (and the correlation
between fixed effects). The fixed effect parameters were chosen to
deliver a decomposition that looks like usual empirical ones, with
a large worker contribution, a large residual variance and limited
firm and sorting contributions. Then, we show the variance decomposition
obtained with the OLS estimates of the fixed effects, and the debiased
variance components (\cite{AndrewsGillSchankUpward2008}).\footnote{We do not consider other bias correction methods (\cite{Gaure2014},
\cite{KlineSaggioSoelvsten2020}, \cite{AzkarateAskasuaZerecero2022})
because the model is homoscedastic.} The row labelled ``ridge'' shows the variance decomposition that
is obtained using the ridge estimator. We display the shares of $\V(W\widehat{\mu})$,
$\V(F\widehat{\phi})$, $2\Cov(W\widehat{\mu},F\widehat{\phi})$.
The sixth column of Table \ref{tab:sim_regression} shows the out-of-sample
MSE. Lastly, we report the ridge regularization parameters obtained
by cross-validation. For cross-validation, we simulate a test sample
as follows. We use the same nodes and communities, but draw a new
$d_{ij}$ for each couple of worker-firm nodes. Then, we draw a new
residual $u_{ijs}$. Again, we keep the largest connected component
for prediction.

The OLS estimation is always strongly biased, but the bias correction
works very well as long as the graph is not too sparse (i.e. for $c=2$
but not for $c=1$). The variance decomposition using the ridge estimator
is a lot less affected by sparsity than OLS. There is a tendency of
the ridge estimator to underestimate the shares of $\V(W\widehat{\mu})$,
$\V(F\widehat{\phi})$, $2\Cov(W\widehat{\mu},F\widehat{\phi})$ and
overestimate the ``residual share''. We report the ridge regularization
parameters as a fraction of the average worker and firm degrees.

\begin{figure}
\begin{centering}
\subfloat[$c=1$: PDF of $\mu_{i}$]{\begin{centering}
\includegraphics[scale=0.35]{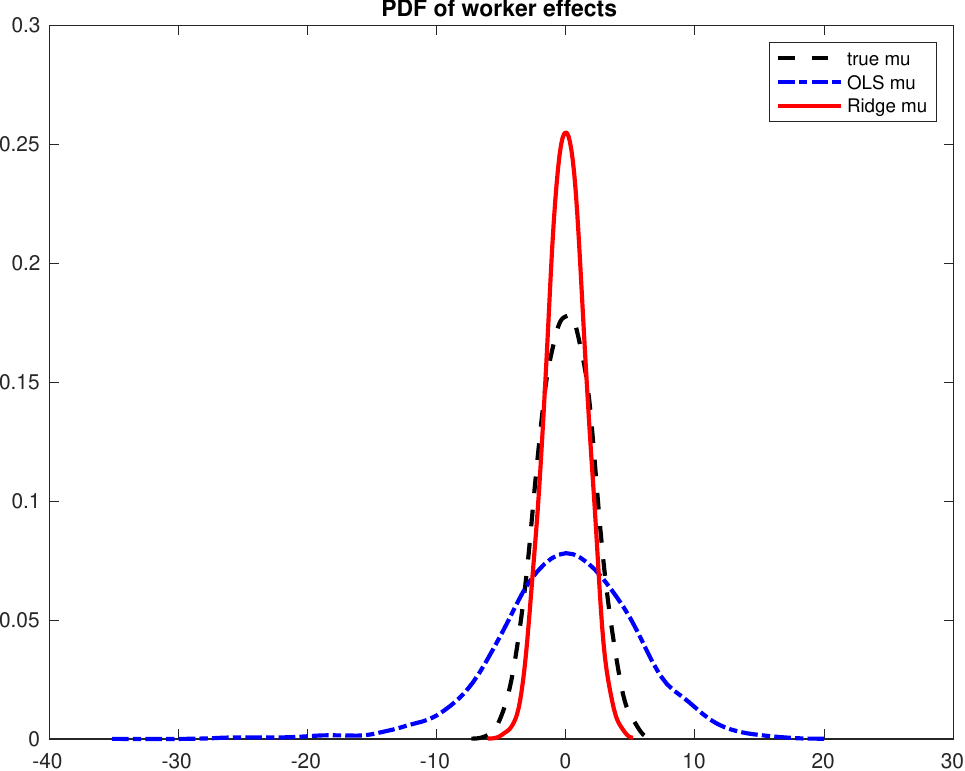} 
\par\end{centering}
}\subfloat[$c=1$: scatterplot of $\mu_{i}$]{\begin{centering}
\includegraphics[scale=0.35]{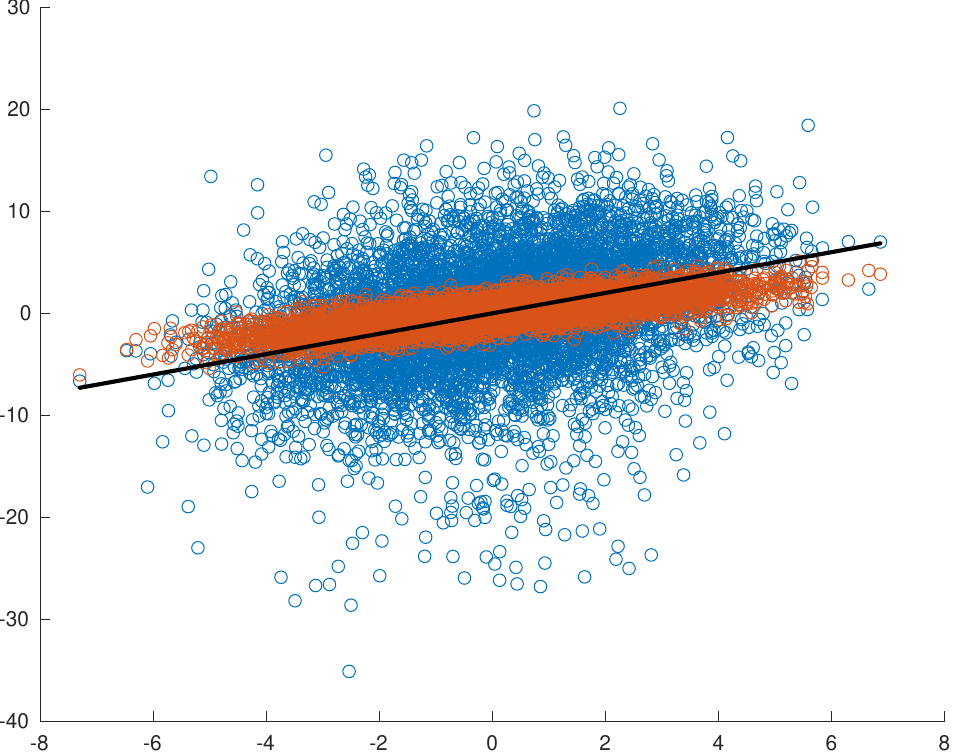} 
\par\end{centering}
}
\par\end{centering}
\begin{centering}
\subfloat[$c=2$: PDF of $\mu_{i}$]{\begin{centering}
\includegraphics[scale=0.35]{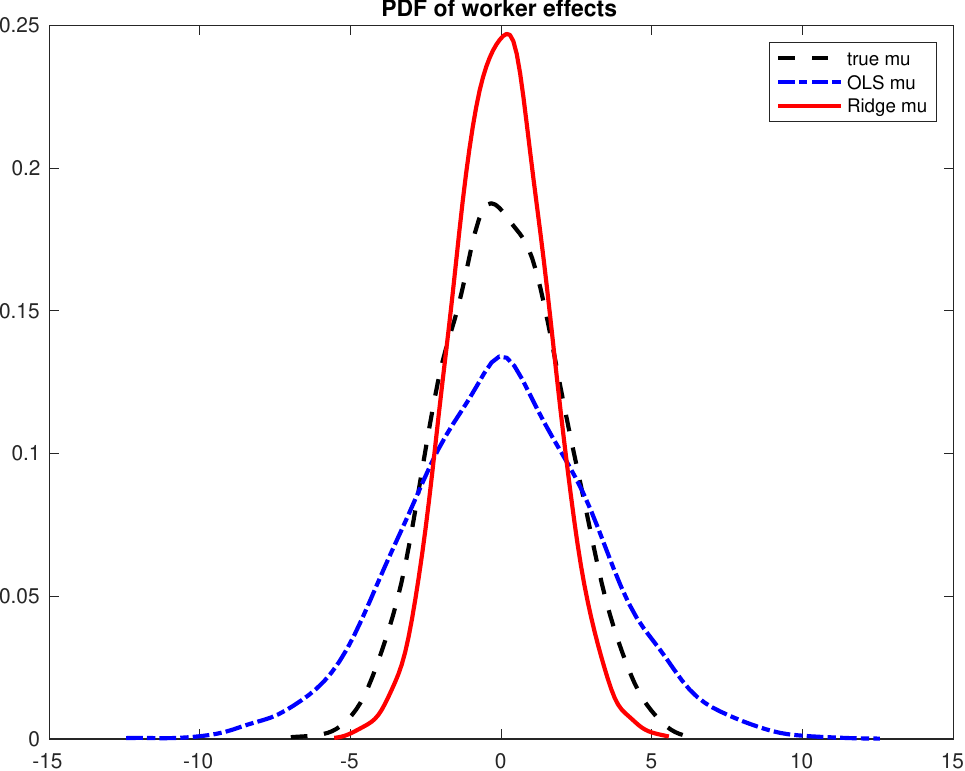} 
\par\end{centering}
}\subfloat[$c=2$: scatterplot of $\mu_{i}$]{\begin{centering}
\includegraphics[scale=0.35]{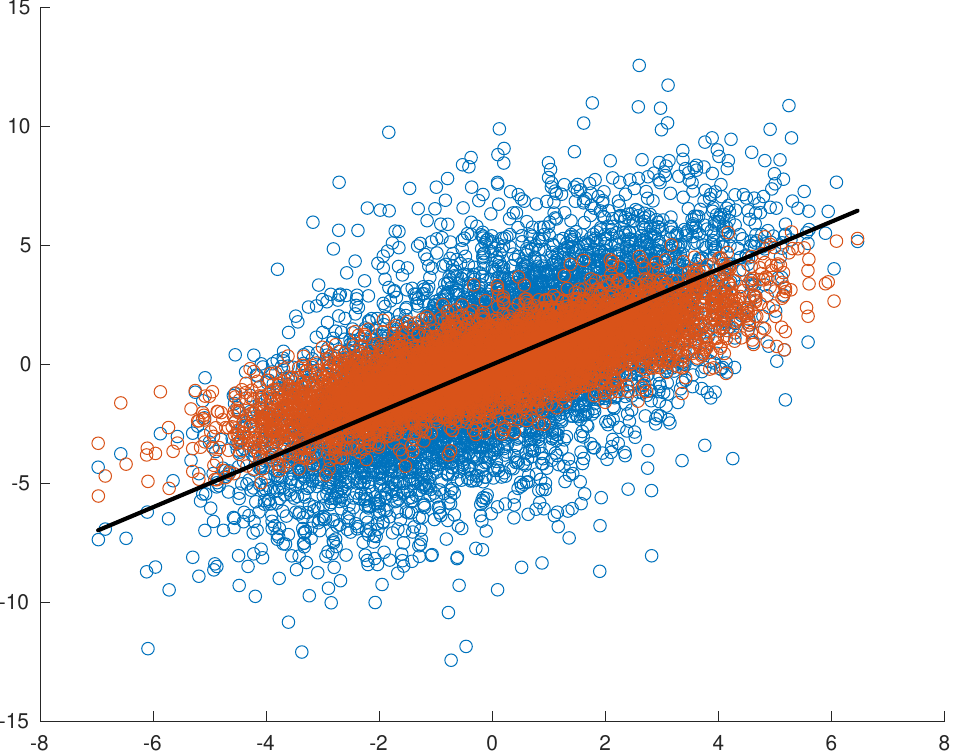} 
\par\end{centering}
}
\par\end{centering}
\caption{Worker fixed effect distributions: true, OLS and ridge}
\label{fig:mu} 
\end{figure}

\begin{figure}
\begin{centering}
\subfloat[$c=1$: PDF of $\phi_{i}$]{\begin{centering}
\includegraphics[scale=0.35]{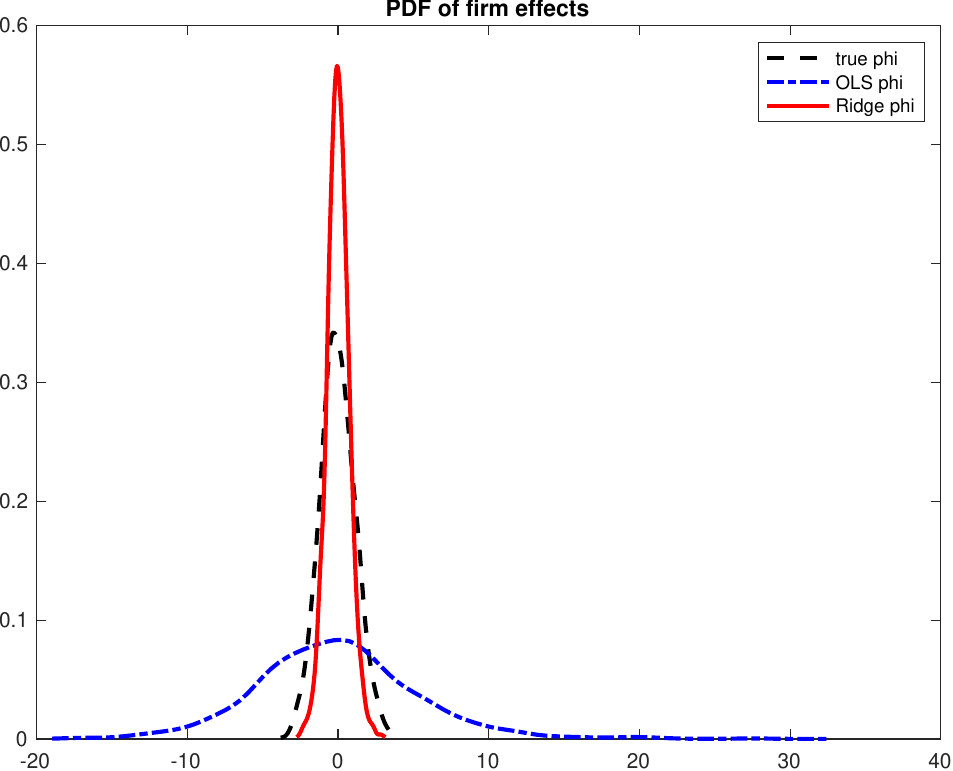} 
\par\end{centering}
}\subfloat[$c=1$: scatterplot of $\phi_{i}$]{\begin{centering}
\includegraphics[scale=0.35]{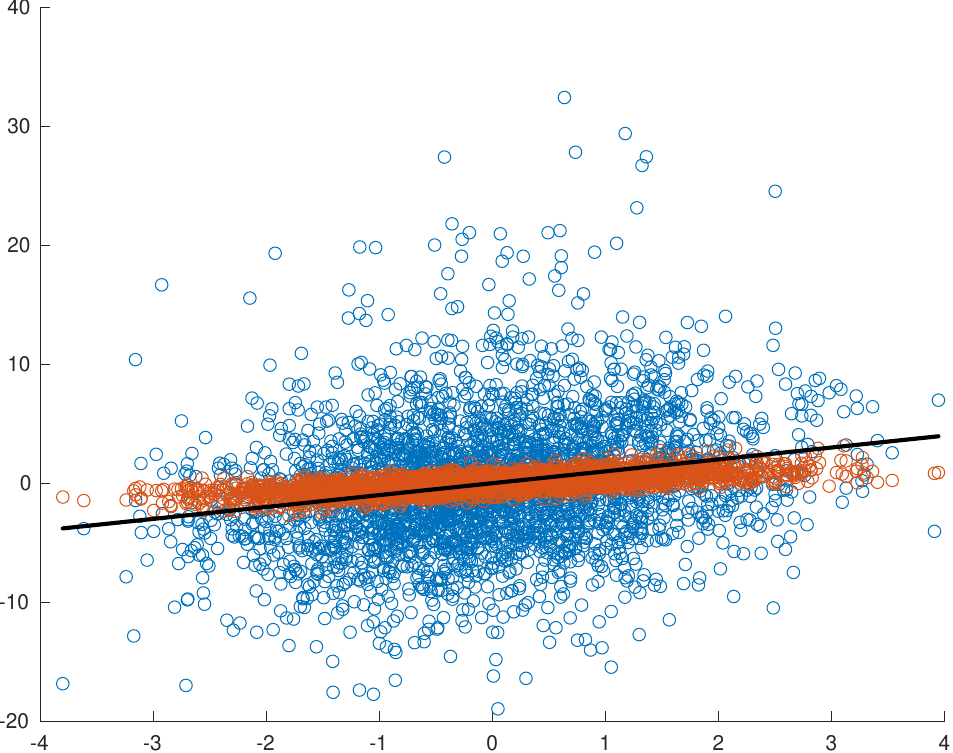} 
\par\end{centering}
}
\par\end{centering}
\begin{centering}
\subfloat[$c=2$: PDF of $\phi_{i}$]{\begin{centering}
\includegraphics[scale=0.35]{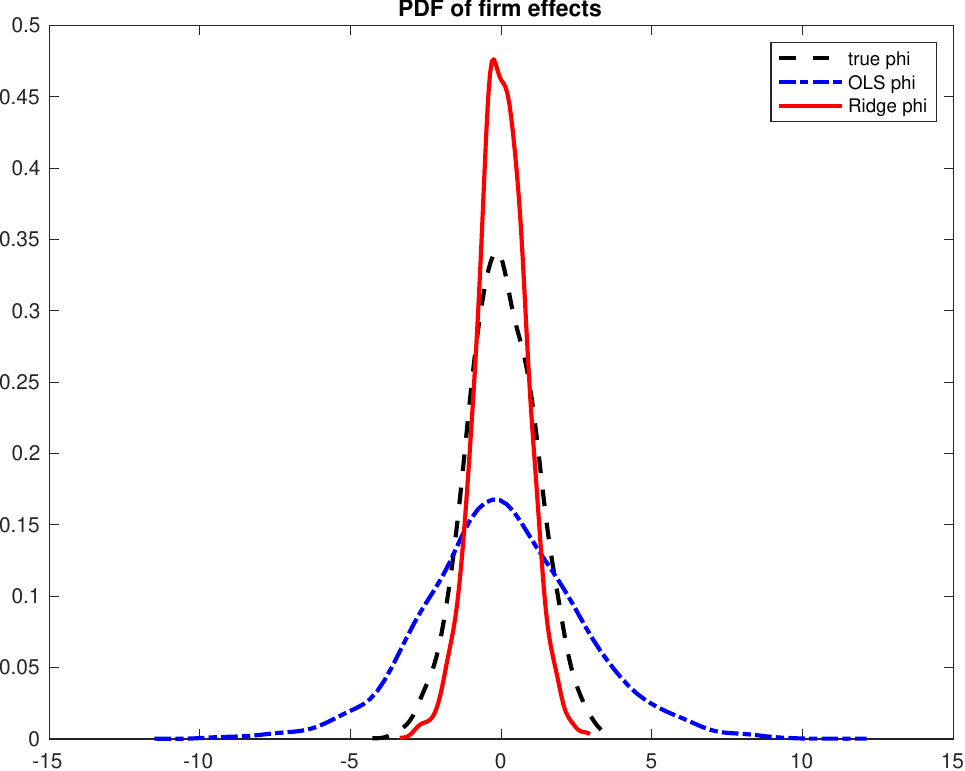} 
\par\end{centering}
}\subfloat[$c=2$: scatterplot of $\phi_{i}$]{\begin{centering}
\includegraphics[scale=0.35]{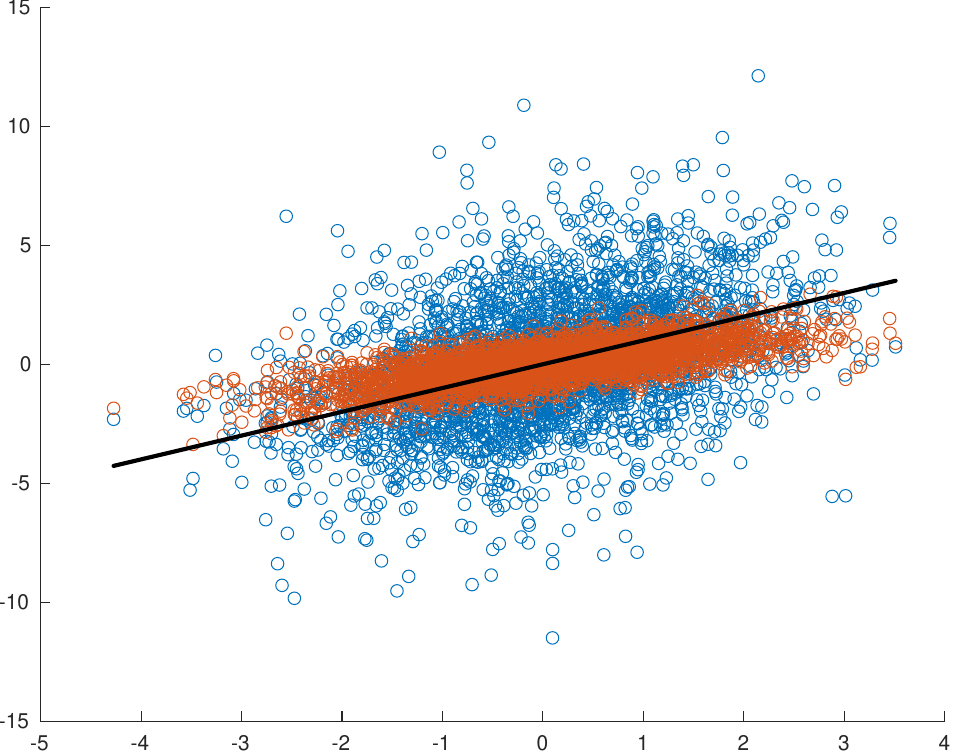} 
\par\end{centering}
}
\par\end{centering}
\caption{Firm fixed effect distributions: true, OLS and ridge}
\label{fig:phi} 
\end{figure}

Next, we turn to the distribution of fixed effects. Figures \ref{fig:mu}
and \ref{fig:phi} show the densities of true and estimated fixed
effects and their scatterplots. The scatterplots show the estimated
versus the true values. The black line is the 45 degree line for a
perfect estimation. We see that OLS estimates are considerably more
dispersed than the true distribution, while the ridge estimator is
only slightly more concentrated. The scatterplots show evidence of
a pattern of the bias for ridge. Large values of the true fixed effects
are associated with a stronger negative bias, and vice versa. 

\subsection{Real data application}

We then propose an application with French matched employer-employee
data (DADS Panel, 1995-2001) similar to the data used in \cite{AbowdKramarzMargolis1999}.
We keep only non agricultural, salaried, private sector employees
working full time, younger than 55. We keep all employment spells,
including those recorded for a fraction of a year. However, we drop
matches with firms with fewer than 50 worker/wage observations and
workers with less that 5 recorded spells. The wage concept is earnings
per day, and we trim wages below the 0.1 percentile and above the
99.9th one. We regress log wages on job tenure, age and age squared,
and use the residuals for our analysis.

After selecting the biggest connected component (out of 2663 ones),
we end up with a network with $n=94491$ workers, $p=3408$ firms
and $N=450125$ wage observations. The average worker degree is $\overline{d}_{w}=4.7637$
and the average firm degree is $\overline{d}_{f}=132.08$. The average
number of wage records per match is 3.5903. The network is very sparse
with 127457 matches, which is 0.03958\% of the $np$ potential ones.

\begin{table}
\setlength{\tabcolsep}{5pt}\caption{Actual data results}

\begin{centering}
\begin{tabular}{lcccc}
\hline 
 & \multicolumn{4}{c}{share of wage variance}\tabularnewline
 & worker  & firm  & 2{*}covariance  & residual \tabularnewline
\hline 
\multicolumn{5}{l}{Largest connected component}\tabularnewline
\hline 
OLS  & 0.9584  & 0.2332  & -0.3163  & 0.1247 \tabularnewline
OLS debiased  & 0.8841  & 0.1912  & -0.2347  & 0.1594 \tabularnewline
ridge  & 0.6005  & 0.0846  & 0.1067  & 0.2083 \tabularnewline
 &  &  &  & \tabularnewline
\multicolumn{5}{l}{Largest strongly connected component }\tabularnewline
\hline 
OLS  & 0.8626  & 0.1008  & -0.0978  & 0.1324 \tabularnewline
OLS debiased  & 0.8153  & 0.0845  & -0.067  & 0.1655 \tabularnewline
Ridge  & 0.5611  & 0.1257  & 0.1001  & 0.2131 \tabularnewline
KSS  & 0.7496  & 0.0655  & 0.0088  & 0.1762 \tabularnewline
\hline 
\end{tabular}
\par\end{centering}
\label{tab:true_regression} 
\end{table}

The log wage residual has a variance of 0.2045 which can be decomposed
as reported in Table \ref{tab:true_regression}. The variance decomposition
obtained using the OLS estimator of the worker and firm effects yields
a rather large negative contribution of the fixed effects covariance
to wages. Biases on plugin variance contributions (assuming homoscedasticity)
are estimated to be small.

Then we show the variance decomposition obtained from ridge regression.\footnote{We show the contributions to the total variance of the worker effect,
the firm effect, 2 times the covariance of worker and firm effects,
and a residual contribution that is the sum of the residual variance
and the covariances between residuals and fixed effects. Contrary
to the OLS estimator, the ridge estimator does not produces orthogonal
residuals.} We estimate the ridge model on the years 1995-2000 and keep the last
year 2001 for cross-validation. 
 By minimizing out-of-sample error, we tune the ridge parameters $\lambda_{w}$
and $\lambda_{f}$ to be 0.055 and 0.060 times the worker and firm
degrees.\footnote{The previous $\log(n+p)$ lower bound is a sufficient condition for our concentration arguments; we do not claim cross-validation yields penalties with this scaling. Instead, we report selected penalties as fractions of average degrees and assess their performance in simulations.}
We now estimate a positive covariance between worker and
firm effects, and with more reasonable contributions of the worker
effects and of the regression residuals.

Next, because we worry about residual heteroscedasticity and autocorrelation,
we aim to compare ridge to the Leave-One-Out (LOO) bias correction
method of \cite{KlineSaggioSoelvsten2020}. The LOO method requires
firms to remain connected after taking one observation out. We use
the strongest version leaving one match out, which asymptotically
removes second-order biases when residuals are heteroscedastic or autoregressive
within matches. OLS and the homoscedastic bias correction are still
very close and deliver a negative covariance term, but of smaller
magnitude, proof that the LOO sample is less sparse. However, ridge
seems little affected by the choice of the sample. The LOO correction
delivers a covariance contribution that is very close to zero, thus
closer to ridge. These estimates are also very close to those in \cite{BabetGodechotPalladino2022}
and \cite{AzkarateAskasuaZerecero2022}, who use exhaustive DADS
data after 2002 (52 million observations) instead of the Panel before
2001, and hourly wages instead of daily wages. Babet et al. use the
LOO correction and Azkarate-Zerecero develop a related parametric bootstrap
correction.

\begin{figure}
\begin{centering}
\includegraphics[clip,scale=0.35]{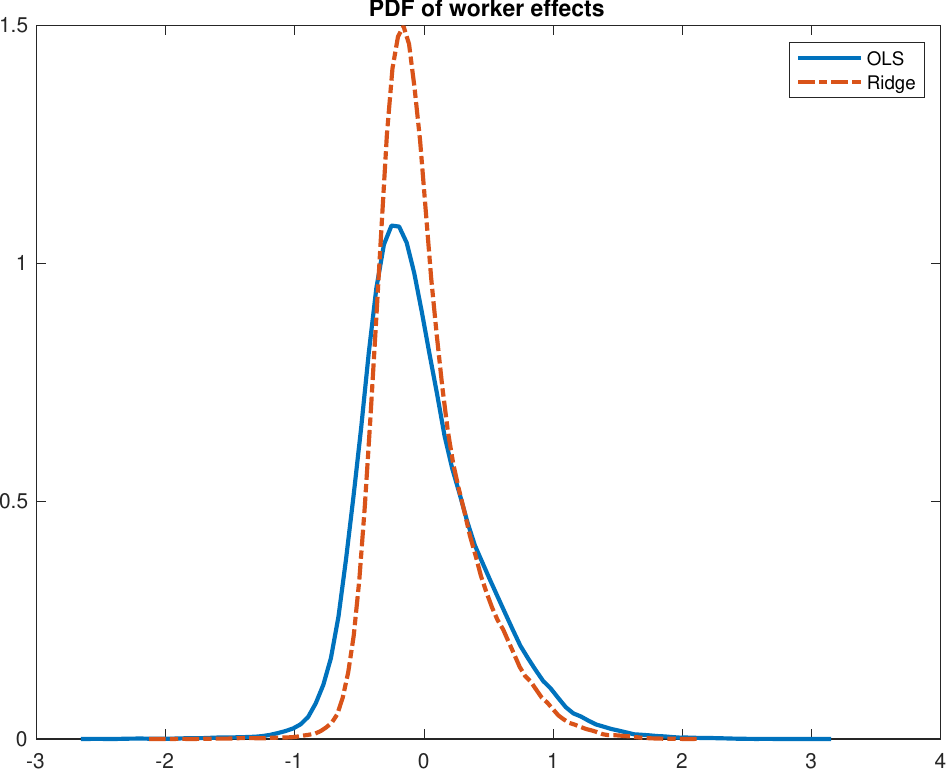}
\includegraphics[scale=0.35]{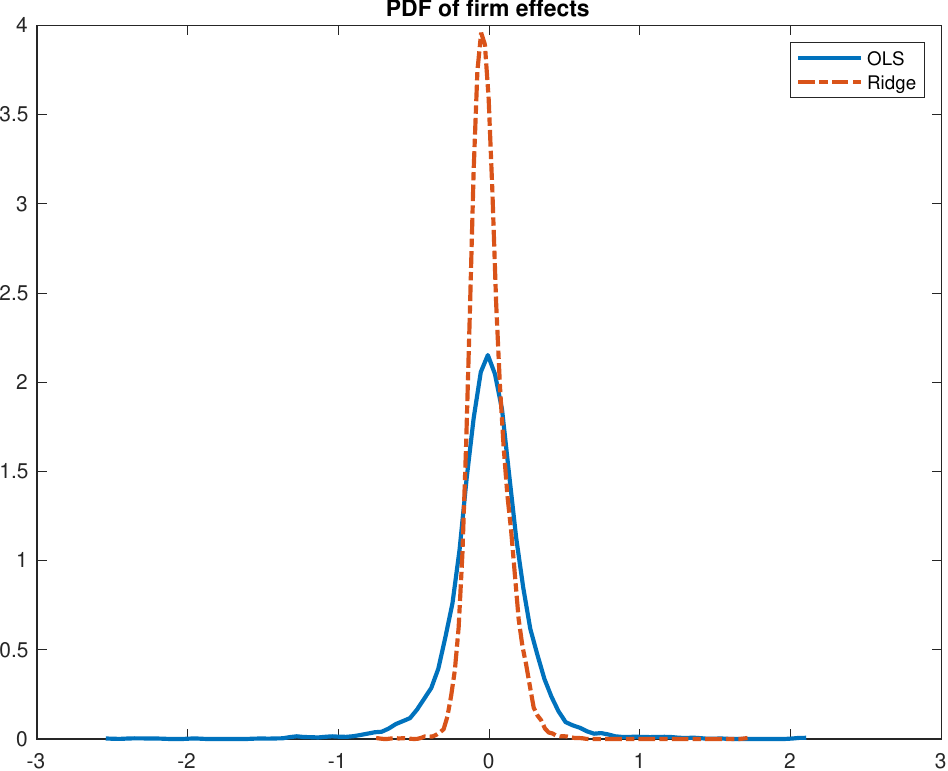} 
\par\end{centering}
\caption{Distributions of estimated fixed effects (OLS and ridge)}
\label{fig:true distributions} 
\end{figure}

\begin{figure}
\begin{centering}
\includegraphics[scale=0.35]{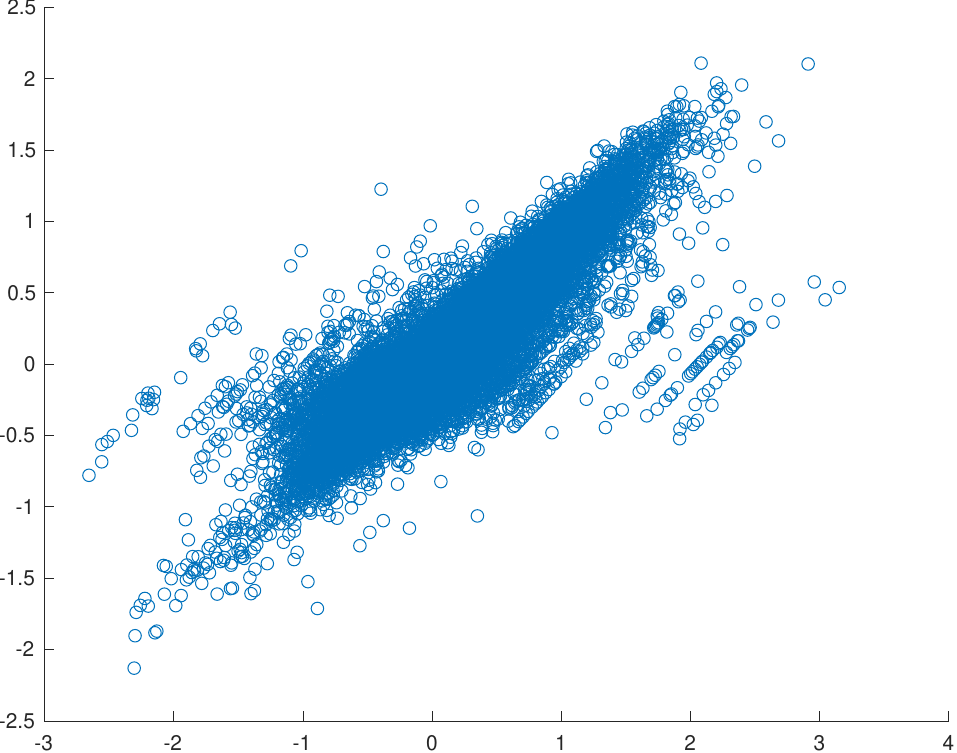}
\includegraphics[scale=0.35]{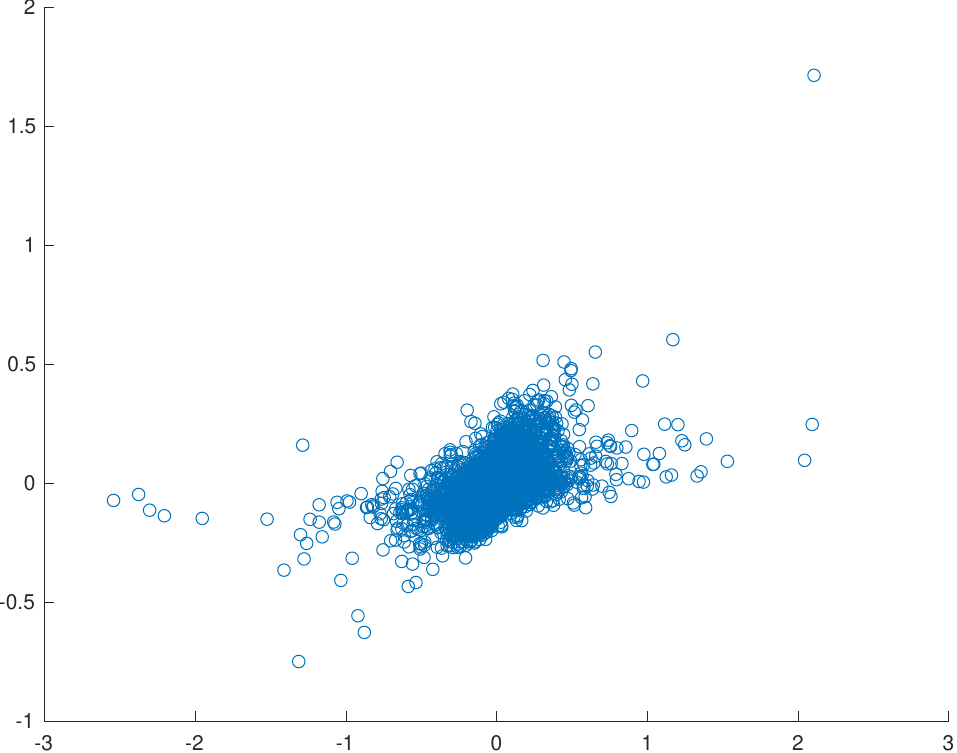} 
\par\end{centering}
\caption{Ridge estimated fixed effects versus OLS (left figure: workers;
right figure: firms)}
\label{fig:ridge vs OLS} 
\end{figure}

Next, we show in Figure \ref{fig:true distributions} the distributions
of the worker and firm fixed effects. The ridge estimator is a lot
more concentrated than the OLS estimator. Figure \ref{fig:ridge vs OLS}
shows each ridge estimate as a function of the corresponding OLS one.
Estimates of firm fixed effects tend to disagree more than estimates
of worker effects.

\section{Conclusion}

In this paper, we have studied ridge penalization as an alternative
to OLS for the estimation of the fixed effects in a two-way fixed
effect regression model like a log wage equation for matched employer-employee
data. We developed a Degree Corrected Stochastic Block model for the
creation of the network. The outcome variable was generated given
the network in a second stage. We derived the asymptotic properties
of the estimator of the fixed effects. We were able to find deterministic
equivalents of the bias and the variance of the estimator. In the
future, further study of these deterministic equivalents will allow
us to characterize the optimal ridge parameter.

\begin{appendix}

\section*{Proofs}\label{app:Proofs}

\subsection{Proof of Lemma \ref{lem:Laplacian eig}}

For any $x=(x_{j})\in\mathbb{R}^{p}$, the expression for $x^{\top}L_{f}x$
equals 
\[
\sum_{j}x_{j}^{2}-\sum_{j,j'}\frac{x_{j}x_{j'}}{\sqrt{d_{\cdot j}d_{\cdot j'}}}\sum_{i}\frac{d_{ij}d_{ij'}}{d_{i\cdot}}=\frac{1}{2}\sum_{i}\frac{1}{d_{i\cdot}}\sum_{j,j'}\left(\frac{x_{j}}{\sqrt{d_{\cdot j}}}-\frac{x_{j'}}{\sqrt{d_{\cdot j'}}}\right)^{2}d_{ij}d_{ij'}\ge0.
\]
Hence $L_{f}$ has nonnegative eigenvalues, and $\lambda_{1}=0$ as
$L_{f}x=\mathbf{0}_{p}$ for $x=\left(\sqrt{d_{\cdot j}}\right)_{j}$.

Then,  
\[
x^{\top}L_{f}x\ge0\Rightarrow1\ge\frac{x^{\top}A_{f}x}{x^{\top}x}.
\]
This Rayleigh quotient gives 1 as an upper bound for $\alpha_{1}$,
and $\alpha_{1}=1$ is indeed the largest eigenvalue of $A_{f}$ with
eigenvector $v_{1}=\left(\sqrt{d_{\cdot j}}\right)$. Finally,
$A_{f}=E^{\top}E$ is obviously positive semidefinite. Its eigenvalues
are nonnegative.

\subsection{Proof of Lemma \ref{lem:max-eig-E'E}}

We have for all $x\in\mathbb{R}^{p}$, 
\begin{multline*}
x^{\top}L_{f,\lambda}x=\frac{1}{2}\sum_{i}\frac{1}{d_{i\cdot}+\lambda_{w}}\sum_{j,j'}\left(\frac{x_{j}}{\sqrt{d_{\cdot j}+\lambda_{f}}}-\frac{x_{j'}}{\sqrt{d_{\cdot j'}+\lambda_{f}}}\right)^{2}d_{ij}d_{ij'}\\
+\sum_{j}\frac{x_{j}^{2}}{d_{\cdot j}+\lambda_{f}}\left(\lambda_{f}+\lambda_{w}\sum_{i}\frac{d_{ij}}{d_{i\cdot}+\lambda_{w}}\right)%
\ge\lambda_{f}\sum_{j}\frac{x_{j}^{2}}{d_{\cdot j}+\lambda_{f}}\ge\frac{\lambda_{f}x^{\top}x}{\max_{j}d_{\cdot j}+\lambda_{f}}.
\end{multline*}

To obtain the second equality, write 
\begin{multline*}
\sum_{i}\frac{1}{d_{i\cdot}+\lambda_{w}}\sum_{j,j'}\left(\frac{x_{j}}{\sqrt{d_{\cdot j}+\lambda_{f}}}-\frac{x_{j'}}{\sqrt{d_{\cdot j'}+\lambda_{f}}}\right)^{2}d_{ij}d_{ij'}\\
=2\sum_{i}\frac{d_{i\cdot}}{d_{i\cdot}+\lambda_{w}}\sum_{j}\frac{x_{j}^{2}}{d_{\cdot j}+\lambda_{f}}d_{ij}-2\sum_{i}\frac{1}{d_{i\cdot}+\lambda_{w}}\sum_{j,j'}\frac{x_{j}}{\sqrt{d_{\cdot j}+\lambda_{f}}}\frac{x_{j'}}{\sqrt{d_{\cdot j'}+\lambda_{f}}}d_{ij}d_{ij'}.
\end{multline*}
Then, $\sum_{i}\frac{d_{i\cdot}}{d_{i\cdot}+\lambda_{w}}\sum_{j}\frac{x_{j}^{2}}{d_{\cdot j}+\lambda_{f}}d_{ij}$
can be written as 
\begin{align*}
\sum_{j}\frac{x_{j}^{2}}{d_{\cdot j}+\lambda_{f}}\sum_{i}\frac{d_{i\cdot}}{d_{i\cdot}+\lambda_{w}}d_{ij}%
  =\sum_{j}x_{j}^{2}-\sum_{j}\frac{x_{j}^{2}}{d_{\cdot j}+\lambda_{f}}\left(\lambda_{f}+\lambda_{w}\sum_{i}\frac{d_{ij}}{d_{i\cdot}+\lambda_{w}}\right).
\end{align*}

\subsection{Proof of Lemma \ref{lem:max-eig-cal(E)'cal(E)}}

For all $x\in\mathbb{R}^{p}$, the inner product $x^{\top}\mathfrak{L}_{f,\lambda}x$
is 
\begin{multline*}
\sum_{j}x_{j}^{2}-\sum_{j,j'}\tfrac{x_{j}}{\sqrt{\theta_{j}C(\cdot,\ell_{j})+\lambda_{f}}}\tfrac{x_{j'}}{\sqrt{\theta_{j'}C(\cdot,\ell_{j'})+\lambda_{f}}}\sum_{k}\tfrac{1}{n_{k}}\tfrac{\theta_{j}C(k,\ell_{j})\theta_{j'}C(k,\ell_{j'})}{C(k,\cdot)/n_{k}+\lambda_{w}}\\
=\tfrac{1}{2}\sum_{k}\tfrac{1/n_{k}}{C(k,\cdot)/n_{k}+\lambda_{w}}\sum_{j,j'}\left(\tfrac{x_{j}}{\sqrt{\theta_{j}C(\cdot,\ell_{j})+\lambda_{f}}}-\tfrac{x_{j'}}{\sqrt{\theta_{j'}C(\cdot,\ell_{j'})+\lambda_{f}}}\right)^{2}\theta_{j}C(k,\ell_{j})\theta_{j'}C(k,\ell_{j'})\\
+\sum_{j}\tfrac{x_{j}^{2}}{\theta_{j}C(\cdot,\ell_{j})+\lambda_{f}}\left(\lambda_{f}+\sum_{k}\tfrac{\lambda_{w}\theta_{j}C(k,\ell_{j})}{C(k,\cdot)/n_{k}+\lambda_{w}}\right)
\ge\lambda_{f}\sum_{j}\tfrac{x_{j}^{2}}{\theta_{j}C(\cdot,\ell_{j})+\lambda_{f}}\ge x^{\top}x\tfrac{\lambda_{f}}{\max_{j}\theta_{j}C(\cdot,\ell_{j})+\lambda_{f}},
\end{multline*}
as
\begin{multline*}
\tfrac{1}{2}\sum_{k}\tfrac{1/n_{k}}{C(k,\cdot)/n_{k}+\lambda_{w}}\sum_{j,j'}\left(\tfrac{x_{j}}{\sqrt{\theta_{j}C(\cdot,\ell_{j})+\lambda_{f}}}-\tfrac{x_{j'}}{\sqrt{\theta_{j'}C(\cdot,\ell_{j'})+\lambda_{f}}}\right)^{2}\theta_{j}C(k,\ell_{j})\theta_{j'}C(k,\ell_{j'})\\
=\sum_{k}\tfrac{\tfrac{1}{n_{k}}C(k,\cdot)}{\tfrac{1}{n_{k}}C(k,\cdot)+\lambda_{w}}\sum_{j}\tfrac{x_{j}^{2}\theta_{j}C(k,\ell_{j})}{\theta_{j}C(\cdot,\ell_{j})+\lambda_{f}}-\sum_{j,j'}\tfrac{x_{j}}{\sqrt{\theta_{j}C(\cdot,\ell_{j})+\lambda_{f}}}\tfrac{x_{j'}}{\sqrt{\theta_{j'}C(\cdot,\ell_{j'})+\lambda_{f}}}\sum_{k}\tfrac{\tfrac{1}{n_{k}}\theta_{j}C(k,\ell_{j})\theta_{j'}C(k,\ell_{j'})}{\frac{1}{n_{k}}C(k,\cdot)+\lambda_{w}}
\end{multline*}
since $\sum_{j}\theta_{j}C(k,\ell_{j})=\sum_{\ell}\sum_{j}\theta_{j}\delta_{\ell_{j}\ell}C(k,\ell)=C(k,\cdot)$, and 
\[
\sum_{k}\tfrac{\tfrac{1}{n_{k}}C(k,\cdot)}{\tfrac{1}{n_{k}}C(k,\cdot)+\lambda_{w}}\sum_{j}\tfrac{x_{j}^{2}\theta_{j}C(k,\ell_{j})}{\theta_{j}C(\cdot,\ell_{j})+\lambda_{f}}=\sum_{j}x_{j}^{2}-\sum_{j}\tfrac{x_{j}^{2}}{\theta_{j}C(\cdot,\ell_{j})+\lambda_{f}}\left(\lambda_{f}+\sum_{k}\tfrac{\lambda_{w}\theta_{j}C(k,\ell_{j})}{\tfrac{1}{n_{k}}C(k,\cdot)+\lambda_{w}}\right).
\]

\subsection{Proof of Theorem \ref{thm:ConcentrationLaplacian}}

Following \cite{ChungRadcliffe2011} (Theorem 2) and \cite{QinRohe2013},
first write 
\begin{align*}
\left\Vert E_{\lambda}-\mathfrak{E}_{\lambda}\right\Vert  & =\left\Vert D_{w,\lambda}^{-1/2}BD_{f,\lambda}^{-1/2}-\mathfrak{D}_{w,\lambda}^{-1/2}\mathfrak{B}\mathfrak{D}_{f,\lambda}^{-1/2}\right\Vert \\
 & =\left\Vert \mathfrak{D}_{w,\lambda}^{-1/2}(B-\mathfrak{B})\mathfrak{D}_{f,\lambda}^{-1/2}+D_{w,\lambda}^{-1/2}BD_{f,\lambda}^{-1/2}-\mathfrak{D}_{w,\lambda}^{-1/2}B\mathfrak{D}_{f,\lambda}^{-1/2}\right\Vert \\
 & \le\left\Vert \mathfrak{D}_{w,\lambda}^{-1/2}(B-\mathfrak{B})\mathfrak{D}_{f,\lambda}^{-1/2}\right\Vert +\left\Vert D_{w,\lambda}^{-1/2}BD_{f,\lambda}^{-1/2}-\mathfrak{D}_{w,\lambda}^{-1/2}B\mathfrak{D}_{f,\lambda}^{-1/2}\right\Vert .
\end{align*}
We then bound each term separately.

\subsubsection*{Bound for $\left\Vert \mathfrak{D}_{w,\lambda}^{-1/2}(B-\mathfrak{B})\mathfrak{D}_{f,\lambda}^{-1/2}\right\Vert $}

Let $X:=\mathfrak{D}_{w,\lambda}^{-1/2}(B-\mathfrak{B})\mathfrak{D}_{f,\lambda}^{-1/2}$.
Write the matrix $X$ as the sum $X=\sum_{i,j}X_{ij}$ with 
\[
X_{ij}=\mathfrak{D}_{w,\lambda}^{-1/2}(d_{ij}-p_{ij})\Delta_{ij}\mathfrak{D}_{f,\lambda}^{-1/2},
\]
where $\Delta_{ij}$ is the $n\times p$ matrix of zeros except in
position $(i,j)$ where there is a one. The matrices $X_{ij}$ are
independent, mean 0 and uniformly bounded. Specifically, we have 
\[
X_{ij}^{\top}X_{ij}=(d_{ij}-p_{ij})^{2}\Delta_{jj}\left(\frac{1}{\left(C(k_{i},\cdot)/n_{k_{i}}+\lambda_{w}\right)\left(\theta_{j}C(\cdot,\ell_{j})+\lambda_{f}\right)}\right)
\]
($\Delta_{ij}^{\top}\Delta_{ij}=\Delta_{jj}$) and $\left\Vert X_{ij}\right\Vert =\sqrt{\eigmax\left(X_{ij}^{\top}X_{ij}\right)}\le\sqrt{M_{w}M_{f}}\le\max(M_{f},M_{w}):=M,$
as $(d_{ij}-p_{ij})^{2}\in\left\{ p_{ij}^{2},(1-p_{ij})^{2}\right\} \le1$,
and denoting $M_{w}=\left(\min_{k}\frac{1}{n_{k}}C(k,\cdot)+\lambda_{w}\right)^{-1}$
and $M_{f}=\left(\min_{j}\theta_{j}C(\cdot,\ell_{j})+\lambda_{f}\right)^{-1}$.
Moreover, let 
\[
v(X)=\max\left\{ \left\Vert \E(X^{\top}X)\right\Vert ,\left\Vert \E(XX^{\top})\right\Vert \right\} =\max\left\{ \left\Vert \sum\nolimits_{ij}\E(X_{ij}^{\top}X_{ij})\right\Vert ,\left\Vert \sum\nolimits_{ij}\E(X_{ij}X_{ij}^{\top})\right\Vert \right\} .
\]
Consider $\sum_{ij}\E(X_{ij}^{\top}X_{ij})=\text{diag}\left(\sum_{i}\frac{p_{ij}(1-p_{ij})}{\left(C(k_{i},\cdot)/n_{k_{i}}+\lambda_{w}\right)\left(\theta_{j}C(\cdot,\ell_{j})+\lambda_{f}\right)}\right)_{j}.$
We have 
\begin{multline*}
\left\Vert \sum_{ij}\E(X_{ij}^{\top}X_{ij})\right\Vert \le\max_{j}\sum_{k}\frac{\theta_{j}C(k,\ell_{j})}{\left(C(k,\cdot)/n_{k}+\lambda_{w}\right)\left(\theta_{j}C(\cdot,\ell_{j})+\lambda_{f}\right)}\\
\le\max_{j}\frac{\theta_{j}C(\cdot,\ell_{j})}{\left(\min_{k}C(k,\cdot)/n_{k}+\lambda_{w}\right)\left(\theta_{j}C(\cdot,\ell_{j})+\lambda_{f}\right)}\le M_{w}.
\end{multline*}
Similarly, $\left\Vert \sum_{ij}\E(X_{ij}X_{ij}^{\top})\right\Vert \le M_{f}.$
We can then apply the matrix Bernstein inequality. For all $t\ge0$,
\begin{align*}
\Pr\left\{ \left\Vert X\right\Vert \ge t\right\}  & \le(n+p)\exp\left(\frac{-t^{2}/2}{v(X)+Mt/3}\right)\le(n+p)\exp\left(\frac{-t^{2}/2}{M+Mt/3}\right).
\end{align*}
We are looking for a bounded $t$, say $0\le t\le1$, such that $\Pr\left\{ \left\Vert X\right\Vert \ge t\right\} \le\epsilon$.
It suffices that 
\[
(n+p)\exp\left(\frac{-t^{2}/2}{M+Mt/3}\right)\le(n+p)\exp\left(\frac{-t^{2}}{3M}\right)=\epsilon,
\]
which holds for $t=\sqrt{3M\ln\frac{n+p}{\epsilon}}.$ Moreover, for
$t\le1$ to hold, we also need $M\le\frac{1}{3\ln\frac{n+p}{\epsilon}}.$

\subsubsection*{Bound for $\left\Vert D_{w,\lambda}^{-1/2}BD_{f,\lambda}^{-1/2}-\mathfrak{D}_{w,\lambda}^{-1/2}B\mathfrak{D}_{f,\lambda}^{-1/2}\right\Vert $}

With $E_{\lambda}=D_{w,\lambda}^{-1/2}BD_{f,\lambda}^{-1/2}$, and
 since $\left\Vert E_{\lambda}\right\Vert <1$ by Lemma \ref{lem:max-eig-E'E},
we have 
\begin{align*}
\left\Vert E_{\lambda}-\mathfrak{D}_{w,\lambda}^{-1/2}B\mathfrak{D}_{f,\lambda}^{-1/2}\right\Vert  & =\left\Vert E_{\lambda}\left(I_{p}-D_{f,\lambda}^{1/2}\mathfrak{D}_{f,\lambda}^{-1/2}\right)+\left(I_{n}-\mathfrak{D}_{w,\lambda}^{-1/2}D_{w,\lambda}^{1/2}\right)E_{\lambda}D_{f,\lambda}^{1/2}\mathfrak{D}_{f,\lambda}^{-1/2}\right\Vert \\
 & \le\left\Vert I_{p}-D_{f,\lambda}^{1/2}\mathfrak{D}_{f,\lambda}^{-1/2}\right\Vert +\left\Vert I_{n}-\mathfrak{D}_{w,\lambda}^{-1/2}D_{w,\lambda}^{1/2}\right\Vert 
\end{align*}
Consider first the term $\left\Vert D_{f,\lambda}^{1/2}\mathfrak{D}_{f,\lambda}^{-1/2}\right\Vert =\max_{j}\sqrt{\frac{d_{\cdot j}+\lambda_{f}}{\theta_{j}C(\cdot,\ell_{j})+\lambda_{f}}}.$
We have 
\[
\sqrt{\frac{d_{\cdot j}+\lambda_{f}}{\theta_{j}C(\cdot,\ell_{j})+\lambda_{f}}}\ge\sqrt{1+t}\Leftrightarrow\frac{d_{\cdot j}-\theta_{j}C(\cdot,\ell_{j})}{\theta_{j}C(\cdot,\ell_{j})+\lambda_{f}}\ge t.
\]
Let $X_{j}=\frac{d_{\cdot j}-\theta_{j}C(\cdot,\ell_{j})}{\theta_{j}C(\cdot,\ell_{j})+\lambda_{f}}=\sum_{i}\frac{d_{ij}-p_{ij}}{\theta_{j}C(\cdot,\ell_{j})+\lambda_{f}}:=\sum_{i}X_{ij}$.
A Bernstein inequality can be applied to each $X_{j}$. First, $\left|X_{ij}\right|\le M_{f}$
and 
\[
v_{j}=\sum_{i}\E\left(X_{ij}^{2}\right)=\left(\tfrac{1}{\theta_{j}C(\cdot,\ell_{j})+\lambda_{f}}\right)^{2}\sum_{i}p_{ij}(1-p_{ij})\le M_{f}^{2}\sum_{i}p_{ij}=M_{f}^{2}\theta_{j}C(\cdot,\ell_{j})\le M_{f}.
\]
Hence, 
\begin{align*}
\Pr\left\{ \frac{d_{\cdot j}-\theta_{j}C(\cdot,\ell_{j})}{\theta_{j}C(\cdot,\ell_{j})+\lambda_{f}}\ge t\right\} \le\exp\left(\frac{-t^{2}/2}{v_{j}+M_{f}t/3}\right)\le\exp\left(\frac{-t^{2}/2}{M_{f}(1+t/3)}\right)%
\le\exp\left(\frac{-t^{2}}{3M}\right)
\end{align*}
if $0\le t\le1$. With the same $t$ as before, we finally obtain
$\Pr\left\{ \frac{d_{\cdot j}-\theta_{j}C(\cdot,\ell_{j})}{\theta_{j}C(\cdot,\ell_{j})+\lambda_{f}}\ge t\right\} \le\frac{\epsilon}{n+p}.$
The same reasoning shows that $\Pr\left\{ -\frac{d_{\cdot j}-\theta_{j}C(\cdot,\ell_{j})}{\theta_{j}C(\cdot,\ell_{j})+\lambda_{f}}\ge t\right\} \le\frac{\epsilon}{n+p}$
and we finally obtain 
\[
\Pr\left\{ \left|\frac{d_{\cdot j}-\theta_{j}C(\cdot,\ell_{j})}{\theta_{j}C(\cdot,\ell_{j})+\lambda_{f}}\right|\ge t\right\} \le\frac{2\epsilon}{n+p}.
\]
We also deduce from the independence of the degrees $d_{ij}$ that
\begin{align*}
\Pr\left\{ \left\Vert D_{f,\lambda}^{1/2}\mathfrak{D}_{f,\lambda}^{-1/2}\right\Vert \ge\sqrt{1+t}\right\}  & =\Pr\left\{ \max_{j}\frac{d_{\cdot j}+\lambda_{f}}{\theta_{j}C(\cdot,\ell_{j})+\lambda_{f}}\ge1+t\right\} \\
 & \le\sum_{j}\Pr\left\{ \frac{d_{\cdot j}-\theta_{j}C(\cdot,\ell_{j})}{\theta_{j}C(\cdot,\ell_{j})+\lambda_{f}}\ge t\right\} =p\frac{\epsilon}{n+p}.
\end{align*}
Next, for $a\in[0,1]$, we can write 
\begin{multline*}
\Pr\left\{ \left\Vert I_{p}-D_{f,\lambda}^{1/2}\mathfrak{D}_{f,\lambda}^{-1/2}\right\Vert \ge a\right\} =\Pr\left\{ \max_{j}\left|\sqrt{\frac{d_{\cdot j}+\lambda_{f}}{\theta_{j}C(\cdot,\ell_{j})+\lambda_{f}}}-1\right|\ge a\right\} \\
=\Pr\left\{ \exists j,\frac{d_{\cdot j}-\theta_{j}C(\cdot,\ell_{j})}{\theta_{j}C(\cdot,\ell_{j})+\lambda_{f}}\ge a^{2}+2a\text{ or }\le a^{2}-2a\right\} \\
\le\sum_{j}\Pr\left\{ \left|\frac{d_{\cdot j}-\theta_{j}C(\cdot,\ell_{j})}{\theta_{j}C(\cdot,\ell_{j})+\lambda_{f}}\right|\ge2a-a^{2}\right\} \le p\frac{2\epsilon}{n+p}
\end{multline*}
if we further let $2a-a^{2}=t$ or $a=1-\sqrt{1-t}\in[0,t]$.

Similar bounds can be obtained for worker degrees: 
\[
\Pr\left\{ \left\Vert D_{w,\lambda}^{1/2}\mathfrak{D}_{w,\lambda}^{-1/2}\right\Vert \ge\sqrt{1+t}\right\} \ge n\frac{\epsilon}{n+p}\ \ \text{ and }\ \ \Pr\left\{ \left\Vert I_{n}-D_{w,\lambda}^{1/2}\mathfrak{D}_{w,\lambda}^{-1/2}\right\Vert \ge a\right\} \le n\frac{2\epsilon}{n+p}.
\]

Hence, with probability at least $1-\frac{2p\epsilon}{n+p}-\frac{p\epsilon}{n+p}-\frac{2n\epsilon}{n+p}=1-\frac{2+3\gamma}{1+\gamma}\epsilon$,\footnote{First, $P\left(A_{k}>a_{k}\right)\le\alpha_{k}$ for some $k$ implies
that $P\left(\bigcap_{k}\left\{ A_{k}\le a_{k}\right\} \right)=1-P\left(\bigcup_{k}\left\{ A_{k}>a_{k}\right\} \right)\ge1-\sum_{k}\alpha_{k}$.
Second, $A_{k}\le a_{k},\forall k$, implies that $\prod_{k}A_{k}\le\prod_{k}a_{k}$.
Hence, $P\left(\prod_{k}A_{k}\le\prod_{k}a_{k}\right)\ge P\left(\bigcap_{k}\left\{ A_{k}\le a_{k}\right\} \right)\ge1-\sum_{k}\alpha_{k}$.}
\[
\left\Vert I_{p}-D_{f,\lambda}^{1/2}\mathfrak{D}_{f,\lambda}^{-1/2}\right\Vert +\left\Vert I_{n}-\mathfrak{D}_{w,\lambda}^{-1/2}D_{w,\lambda}^{1/2}\right\Vert \left\Vert D_{f,\lambda}^{1/2}\mathfrak{D}_{f,\lambda}^{-1/2}\right\Vert \le\left(1-\sqrt{1-t}\right)\left(1+\sqrt{t+1}\right)\le3t,
\]
as $1-\sqrt{1-t}\le t$ if $t\in[0,1]$, and $1+\sqrt{t+1}<3$. Thus,
with probability at least $1-\frac{2+3\gamma}{1+\gamma}\epsilon-\epsilon=1-\frac{3+4\gamma}{1+\gamma}\epsilon$
we have $\left\Vert E_{\lambda}-\mathfrak{E}_{\lambda}\right\Vert \le4t$.

\subsubsection*{Bound for the Laplacians}

The inequality follows for the Laplacians with a factor 2, as we have
\begin{align*}
\left\Vert L_{f,\lambda}-\mathfrak{L}_{f,\lambda}\right\Vert  & =\left\Vert \mathfrak{E}_{\lambda}^{\top}\mathcal{\mathfrak{E}{}_{\lambda}}-E_{\lambda}^{\top}E_{\lambda}\right\Vert =\left\Vert \mathfrak{E}_{\lambda}^{\top}(\mathfrak{E}_{\lambda}-E_{\lambda})+(\mathfrak{E}_{\lambda}-E_{\lambda})^{\top}E_{\lambda}\right\Vert \\
 & \le\left(\left\Vert \mathfrak{E}_{\lambda}\right\Vert +\left\Vert E_{\lambda}\right\Vert \right)\left\Vert \mathfrak{E}_{\lambda}-E_{\lambda}\right\Vert \le2\left\Vert \mathfrak{E}_{\lambda}-E_{\lambda}\right\Vert ,
\end{align*}
and $\left\Vert L_{w,\lambda}-\mathfrak{L}_{w,\lambda}\right\Vert \le2\left\Vert \mathfrak{E}_{\lambda}-E_{\lambda}\right\Vert .$

\subsection{Proof of Theorem \ref{thm:ConcentrationInverseLaplacian}}

We have 
\begin{align*}
\left\Vert L_{f,\lambda}-\mathfrak{L}_{f,\lambda}\right\Vert  & =\left\Vert \mathfrak{E}_{\lambda}^{\top}\mathfrak{E}_{\lambda}-E_{\lambda}^{\top}E_{\lambda}\right\Vert =\left\Vert \mathfrak{E}_{\lambda}^{\top}(\mathfrak{E}_{\lambda}-E_{\lambda})+(\mathfrak{E}_{\lambda}-E_{\lambda})^{\top}E_{\lambda}\right\Vert \\
 & \le\left(\left\Vert \mathfrak{E}_{\lambda}\right\Vert +\left\Vert E_{\lambda}\right\Vert \right)\left\Vert \mathfrak{E}_{\lambda}-E_{\lambda}\right\Vert \le2\left\Vert \mathfrak{E}_{\lambda}-E_{\lambda}\right\Vert ,
\end{align*}
as $\left\Vert \mathfrak{E}_{\lambda}\right\Vert ,\left\Vert E_{\lambda}\right\Vert \le1$.
Then, 
\begin{align*}
\left\Vert L_{f,\lambda}^{-1}-\mathfrak{L}_{f,\lambda}^{-1}\right\Vert  & =\left\Vert L_{f,\lambda}^{-1}\left(L_{f,\lambda}-\mathfrak{L}_{f,\lambda}\right)\mathfrak{L}_{f,\lambda}^{-1}\right\Vert \le\left\Vert L_{f,\lambda}^{-1}\right\Vert \left\Vert \mathfrak{L}_{f,\lambda}^{-1}\right\Vert \left\Vert L_{f,\lambda}-\mathfrak{L}_{f,\lambda}\right\Vert \\
 & \le\frac{\max_{j}d_{\cdot j}+\lambda_{f}}{\lambda_{f}}\frac{\max_{j}\theta_{j}C(\cdot,\ell_{j})+\lambda_{f}}{\lambda_{f}}2\left\Vert \mathfrak{E}_{\lambda}-E_{\lambda}\right\Vert, 
\end{align*}
as by Lemma \ref{lem:max-eig-E'E}, $\eigmin(L_{f,\lambda})\ge\frac{\lambda_{f}}{\max_{j}d_{\cdot j}+\lambda_{f}}$,
and therefore $\eigmax(L_{f,\lambda}^{-1})\le\frac{\max_{j}d_{\cdot j}+\lambda_{f}}{\lambda_{f}}$,
and by Lemma \ref{lem:max-eig-cal(E)'cal(E)}, $\eigmin(\mathfrak{L}_{f,\lambda})\ge\frac{\lambda_{f}}{\max_{j}\theta_{j}C(\cdot,\ell_{j})+\lambda_{f}}=\frac{\lambda_{f}}{\overline{\delta}_{f}+\lambda_{f}}$.
Now, we have shown that, for $M_{w}\vee M_{f}\le\frac{1}{3\ln\frac{n+p}{\epsilon}}$ 
and $t=\sqrt{3(M_{w}\vee M_{f})\ln\frac{n+p}{\epsilon}}\le1$, for all $j$,
\[
\Pr\left\{ \frac{d_{\cdot j}+\lambda_{f}}{\theta_{j}C(\cdot,\ell_{j})+\lambda_{f}}\ge t+1\right\} \le\frac{\epsilon}{n+p}.
\]
This implies that 
\[
\Pr\left\{ \frac{d_{\cdot j}+\lambda_{f}}{\max_{j}\theta_{j}C(\cdot,\ell_{j})+\lambda_{f}}\ge t+1\right\} \le\frac{\epsilon}{n+p}.
\]
And therefore, 
\[
\Pr\left\{ \frac{\max_{j}d_{\cdot j}+\lambda_{f}}{\max_{j}\theta_{j}C(\cdot,\ell_{j})+\lambda_{f}}\ge t+1\right\} \le\frac{p\epsilon}{n+p}=\frac{\gamma}{1+\gamma}\epsilon.
\]
(See the proof of Theorem \ref{thm:ConcentrationLaplacian}.)
We also have 
\[
\Pr\left\{ \left\Vert E_{\lambda}-\mathfrak{E}_{\lambda}\right\Vert \ge4t\right\} \le\frac{3+4\gamma}{1+\gamma}\epsilon.
\]
Therefore it holds with probability at least $1-\frac{3+4\gamma}{1+\gamma}\epsilon-\frac{\gamma}{1+\gamma}\epsilon=1-\frac{3+5\gamma}{1+\gamma}\epsilon$ that 
\[
\left\Vert L_{f,\lambda}^{-1}-\mathfrak{L}_{f,\lambda}^{-1}\right\Vert \le8t(t+1)\left(\frac{\max_{j}\theta_{j}C(\cdot,\ell_{j})+\lambda_{f}}{\lambda_{f}}\right)^{2}\le16t\left(\frac{\max_{j}\theta_{j}C(\cdot,\ell_{j})+\lambda_{f}}{\lambda_{f}}\right)^{2},
\]
since $t\le1$. 

Similarly, we can prove that 
\[
\Pr\left\{ \left\Vert D_{w,\lambda}^{1/2}\mathfrak{D}_{w,\lambda}^{-1/2}\right\Vert \ge\sqrt{1+t}\right\} \leq\Pr\left\{ \frac{\max_{i}d_{i\cdot}+\lambda_{w}}{\max_{k}\frac{1}{n_{k}}C(k,\cdot)+\lambda_{w}}\ge1+t\right\} \leq\frac{n}{n+p}\epsilon=\frac{1}{1+\gamma}\epsilon.
\]
Therefore, with probability at least $1-\frac{3+4\gamma}{1+\gamma}\epsilon-\frac{1}{1+\gamma}\epsilon=1-4\epsilon$,
we have 
\[
\left\Vert L_{w,\lambda}^{-1}-\mathfrak{L}_{w,\lambda}^{-1}\right\Vert \le16t\left(\frac{\max_{k}\frac{1}{n_{k}}C(k,\cdot)+\lambda_{w}}{\lambda_{w}}\right)^{2}.
\]

\subsection{Proof of Theorem \ref{thm:ConcentrationInvUnnormalizedLap}}

We have 
\[
\left\Vert \widetilde{L}_{f,\lambda}^{-1}-\widetilde{\mathfrak{L}}_{f,\lambda}^{-1}\right\Vert \le\left\Vert \mathfrak{D}_{f,\lambda}^{-1/2}\left(L_{f,\lambda}^{-1}-\mathfrak{L}_{f,\lambda}^{-1}\right)\mathfrak{D}_{f,\lambda}^{-1/2}\right\Vert +\left\Vert \widetilde{L}_{f,\lambda}^{-1}-\mathfrak{D}_{f,\lambda}^{-1/2}L_{f,\lambda}^{-1}\mathfrak{D}_{f,\lambda}^{-1/2}\right\Vert .
\]
The first term is bounded as 
\[
\left\Vert \mathfrak{D}_{f,\lambda}^{-1/2}\left(L_{f,\lambda}^{-1}-\mathfrak{L}_{f,\lambda}^{-1}\right)\mathfrak{D}_{f,\lambda}^{-1/2}\right\Vert \le\left\Vert \mathfrak{D}_{f,\lambda}^{-1}\right\Vert \left\Vert L_{f,\lambda}^{-1}-\mathfrak{L}_{f,\lambda}^{-1}\right\Vert \le M_{f}\left\Vert L_{f,\lambda}^{-1}-\mathfrak{L}_{f,\lambda}^{-1}\right\Vert .
\]
We can write the second term as
\begin{align*}
\left\Vert \widetilde{L}_{f,\lambda}^{-1}-\mathfrak{D}_{f,\lambda}^{-1/2}L_{f,\lambda}^{-1}\mathfrak{D}_{f,\lambda}^{-1/2}\right\Vert  & =\left\Vert \widetilde{L}_{f,\lambda}^{-1}\left(I_{p}-D_{f,\lambda}^{1/2}\mathfrak{D}_{f,\lambda}^{-1/2}\right)+\left(I_{p}-\mathfrak{D}_{f,\lambda}^{-1/2}D_{f,\lambda}^{1/2}\right)\widetilde{L}_{f,\lambda}^{-1}D_{f,\lambda}^{1/2}\mathfrak{D}_{f,\lambda}^{-1/2}\right\Vert \\
 & \le\left\Vert \widetilde{L}_{f,\lambda}^{-1}\right\Vert \left\Vert I_{p}-D_{f,\lambda}^{1/2}\mathfrak{D}_{f,\lambda}^{-1/2}\right\Vert \left(1+\left\Vert D_{f,\lambda}^{1/2}\mathfrak{D}_{f,\lambda}^{-1/2}\right\Vert \right).
\end{align*}
We have already shown that 
\begin{align*}
\Pr\left\{ \left\Vert L_{f,\lambda}^{-1}-\mathfrak{L}_{f,\lambda}^{-1}\right\Vert >16t\left(\frac{\overline{\delta}_{f}+\lambda_{f}}{\lambda_{f}}\right)^{2}\right\}  & \le\frac{3+5\gamma}{1+\gamma}\epsilon,\\
\Pr\left\{ \frac{\max_{j}d_{\cdot j}+\lambda_{f}}{\overline{\delta}_{f}+\lambda_{f}}>t+1\right\}  & \le\frac{p\epsilon}{n+p}=\frac{\gamma}{1+\gamma}\epsilon,\\
\Pr\left\{ \left\Vert I_{p}-D_{f,\lambda}^{1/2}\mathfrak{D}_{f,\lambda}^{-1/2}\right\Vert >1-\sqrt{1-t}\right\}  & \le\frac{2p\epsilon}{n+p}=\frac{2\gamma}{1+\gamma}\epsilon,\\
\Pr\left\{ \left\Vert D_{f,\lambda}^{1/2}\mathfrak{D}_{f,\lambda}^{-1/2}\right\Vert >\sqrt{1+t}\right\}  & \le\frac{p}{n+p}\epsilon=\frac{\gamma}{1+\gamma}\epsilon,
\end{align*}
with $\overline{\delta}_{f}=\max_{j}\theta_{j}C(\cdot,\ell_{j})$.
And by Lemma \ref{lem:max-eig-E'E},
$\left\Vert L_{f,\lambda}^{-1}\right\Vert \le\frac{\max_{j}d_{\cdot j}+\lambda_{f}}{\lambda_{f}}.$
Hence, with probability at least $1-\frac{\gamma}{1+\gamma}\epsilon,$
\begin{multline*}
\left\Vert \widetilde{L}_{f,\lambda}^{-1}\right\Vert =\left\Vert D_{f,\lambda}^{-1/2}L_{f,\lambda}^{-1}D_{f,\lambda}^{-1/2}\right\Vert \le\left\Vert D_{f,\lambda}^{-1}\right\Vert \left\Vert L_{f,\lambda}^{-1}\right\Vert \\
\le\frac{1}{\lambda_{f}}\frac{\max_{j}d_{\cdot j}+\lambda_{f}}{\lambda_{f}}=\frac{\overline{\delta}_{f}+\lambda_{f}}{\lambda_{f}^{2}}\frac{\max_{j}d_{\cdot j}+\lambda_{f}}{\overline{\delta}_{f}+\lambda_{f}}\le\frac{\overline{\delta}_{f}+\lambda_{f}}{\lambda_{f}^{2}}(t+1)\le2\frac{\overline{\delta}_{f}+\lambda_{f}}{\lambda_{f}^{2}}.
\end{multline*}
And with probability at least 
\[
1-\frac{3+5\gamma}{1+\gamma}\epsilon-\frac{(1+2+1)\gamma}{1+\gamma}\epsilon=1-\frac{3+9\gamma}{1+\gamma}\epsilon,
\]
we can bound 
\[
\left\Vert \widetilde{L}_{f,\lambda}^{-1}-\widetilde{\mathfrak{L}}_{f,\lambda}^{-1}\right\Vert \le M_{f}16t\left(\frac{\overline{\delta}_{f}+\lambda_{f}}{\lambda_{f}}\right)^{2}+\frac{1}{\lambda_{f}}\frac{\overline{\delta}_{f}+\lambda_{f}}{\lambda_{f}}5t=\left(16\frac{\overline{\delta}_{f}+\lambda_{f}}{\underline{\delta}_{f}+\lambda_{f}}+5\right)\frac{\overline{\delta}_{f}+\lambda_{f}}{\lambda_{f}^{2}}t.
\]

Similarly, with symmetric notations, with probability at least $1-\frac{1}{1+\gamma}\epsilon$,
$\left\Vert \widetilde{L}_{w,\lambda}^{-1}\right\Vert \le2\frac{\overline{\delta}_{w}+\lambda_{w}}{\lambda_{w}^{2}}$,
with $\overline{\delta}_{w}=\max_{k}\frac{1}{n_{k}}C(k,\cdot)$. And
with probability at least $1-4\epsilon-\frac{4}{1+\gamma}\epsilon=1-\frac{8+4\gamma}{1+\gamma}\epsilon$, we
have
\[
\left\Vert \widetilde{L}_{w,\lambda}^{-1}-\widetilde{\mathfrak{L}}_{w,\lambda}^{-1}\right\Vert \le\left(16\frac{\overline{\delta}_{w}+\lambda_{w}}{\underline{\delta}_{w}+\lambda_{w}}+5\right)\frac{\overline{\delta}_{w}+\lambda_{w}}{\lambda_{w}^{2}}t,
\]

\subsection{Proof of Theorem \ref{thm:ConcentrationRidgeBias}}

We start by proving two intermediate inequalities.

\paragraph*{Concentration bound for $\left\Vert BD_{f,\lambda}^{-1}-\mathfrak{B}\mathfrak{D}_{f,\lambda}^{-1}\right\Vert $ and
$\left\Vert BD_{w,\lambda}^{-1}-\mathfrak{B}\mathfrak{D}_{w,\lambda}^{-1}\right\Vert $}

We can bound 
\begin{align*}
\left\Vert BD_{f,\lambda}^{-1}-\mathfrak{B}\mathfrak{D}_{f,\lambda}^{-1}\right\Vert  & =\left\Vert (B-\mathfrak{B})\mathfrak{D}_{f,\lambda}^{-1}+B\left(D_{f,\lambda}^{-1}-\mathfrak{D}_{f,\lambda}^{-1}\right)\right\Vert \\
 & \le\left\Vert (B-\mathfrak{B})\mathfrak{D}_{f,\lambda}^{-1}\right\Vert +\left\Vert BD_{f,\lambda}^{-1}-B\mathfrak{D}_{f,\lambda}^{-1}\right\Vert .
\end{align*}
First, with $E_{\lambda}=D_{w,\lambda}^{-1/2}BD_{f,\lambda}^{-1/2}$,
\begin{align*}
\left\Vert BD_{f,\lambda}^{-1}-B\mathfrak{D}_{f,\lambda}^{-1}\right\Vert  & =\left\Vert BD_{f,\lambda}^{-1}\left(I_{p}-D_{f,\lambda}\mathfrak{D}_{f,\lambda}^{-1}\right)\right\Vert 
 \le\left\Vert BD_{f,\lambda}^{-1}\right\Vert \left\Vert I_{p}-D_{f,\lambda}\mathfrak{D}_{f,\lambda}^{-1}\right\Vert \\
 & \le\left\Vert D_{w,\lambda}^{1/2}E_{\lambda}D_{f,\lambda}^{-1/2}\right\Vert \left\Vert I_{p}-D_{f,\lambda}\mathfrak{D}_{f,\lambda}^{-1}\right\Vert \\
 & \le\left\Vert D_{w,\lambda}^{1/2}\mathfrak{D}_{w,\lambda}^{-1/2}\right\Vert \left\Vert \mathfrak{D}_{w,\lambda}^{1/2}E_{\lambda}D_{f,\lambda}^{-1/2}\right\Vert \left\Vert I_{p}-D_{f,\lambda}\mathfrak{D}_{f,\lambda}^{-1}\right\Vert ,
\end{align*}
where 
\[
\left\Vert \mathfrak{D}_{w,\lambda}^{1/2}E_{\lambda}D_{f,\lambda}^{-1/2}\right\Vert \le\left\Vert \mathfrak{D}_{w,\lambda}^{1/2}\right\Vert \left\Vert E_{\lambda}\right\Vert \left\Vert D_{f,\lambda}^{-1/2}\right\Vert \le\sqrt{\frac{\overline{\delta}_{w}+\lambda_{w}}{\lambda_{f}}},
\]
as $\left\Vert E_{\lambda}\right\Vert \le1$. The other terms can
be bounded as in the previous theorems:
\begin{align*}
\Pr\left\{ \left\Vert I_{p}-D_{f,\lambda}^{1/2}\mathfrak{D}_{f,\lambda}^{-1/2}\right\Vert \ge1-\sqrt{1-t}\right\}  & \le\frac{2p\epsilon}{n+p}=\frac{2\gamma}{1+\gamma}\epsilon,\\
\Pr\left\{ \left\Vert D_{w,\lambda}^{1/2}\mathfrak{D}_{w,\lambda}^{-1/2}\right\Vert \ge\sqrt{1+t}\right\}  & \le\frac{n}{n+p}\epsilon=\frac{1}{1+\gamma}\epsilon.
\end{align*}
Second, we turn to the term $\left\Vert (B-\mathfrak{B})\mathfrak{D}_{f,\lambda}^{-1}\right\Vert $
that we bound as in Theorem \ref{thm:ConcentrationLaplacian}. Let
\[
X_{ij}=\frac{d_{ij}-p_{ij}}{\theta_{j}C(\cdot,\ell_{j})+\lambda_{j}}\Delta_{ij}.
\]
Then, $\left\Vert X_{ij}\right\Vert \le M_{f}$, and
\begin{multline*}
\left\Vert \sum\nolimits_{ij}\E(X_{ij}^{\top}X_{ij})\right\Vert =\max_{j}\sum_{i}\frac{p_{ij}(1-p_{ij})}{\left(\theta_{j}C(\cdot,\ell_{j})+\lambda_{f}\right)^{2}}\le\max_{j}\sum_{i}\frac{p_{ij}}{\left(\theta_{j}C(\cdot,\ell_{j})+\lambda_{f}\right)^{2}}\\
=\max_{j}\frac{\theta_{j}C(\cdot,\ell_{j})}{\left(\theta_{j}C(\cdot,\ell_{j})+\lambda_{f}\right)^{2}}\le M_{f},
\end{multline*}
while
\begin{multline*}
\left\Vert \sum\nolimits_{ij}\E(X_{ij}X_{ij}^{\top})\right\Vert =\max_{i}\sum_{j}\frac{p_{ij}(1-p_{ij})}{\left(\theta_{j}C(\cdot,\ell_{j})+\lambda_{f}\right)^{2}}\le\max_{i}\sum_{j}\frac{p_{ij}}{\left(\theta_{j}C(\cdot,\ell_{j})+\lambda_{f}\right)^{2}}\\
=\max_{k}\sum_{j}\frac{\frac{1}{n_{k}}\theta_{j}C(k,\ell_{j})}{\left(\theta_{j}C(\cdot,\ell_{j})+\lambda_{f}\right)^{2}}\le M_{f}^{2}\max_{k}\sum_{j}\frac{1}{n_{k}}\theta_{j}C(k,\ell_{j})=M_{f}^{2}\max_{k}\frac{1}{n_{k}}C(k,\cdot).
\end{multline*}
We can then apply the matrix Bernstein inequality. For all $t_{f}\ge0$,
\begin{align*}
\Pr\left\{ \left\Vert (B-\mathfrak{B})\mathfrak{D}_{f,\lambda}^{-1}\right\Vert \ge t_{f}\right\}  & \le(n+p)\exp\left(\frac{-t_{f}^{2}/2}{M_{f}\chi_{f}+M_{f}t_{f}/3}\right)
\end{align*}
where $\chi_{f}=\max\left(1,M_{f}\overline{\delta}_{w}\right)=\max\left(1,\frac{\overline{\delta}_{w}}{\underline{\delta}_{f}+\lambda_{f}}\right)$.
Let
\begin{align*}
t_{f} & =\frac{M_{f}}{3}\ln\frac{n+p}{\epsilon}+\sqrt{\left(\frac{M_{f}}{3}\ln\frac{n+p}{\epsilon}\right)^{2}+2\chi_{f}M_{f}\ln\frac{n+p}{\epsilon}}
\end{align*}
be the positive root of 
\[
(n+p)\exp\left(\frac{-t_{f}^{2}/2}{M_{f}\chi_{f}+M_{f}t_{f}/3}\right)=\epsilon.
\]
Since $t=\sqrt{3M\ln\frac{n+p}{\epsilon}}\le1$ (so $t^{4}\le t^{2}\le t$)
and since $\chi_{f}\ge1$ (so $\chi_{f}\ge\sqrt{\chi_{f}})$, we have
\[
t_{f}\le\frac{t^{2}}{9}+\sqrt{\frac{t^{4}}{9^{2}}+\frac{2t^{2\chi_{f}}}{3}}\le\frac{t}{9}\left(1+\sqrt{1+6\chi_{f}}\right)\le t\sqrt{\chi_{f}}.
\]
Finally, with probability at least 
\[
1-\epsilon-\frac{1+2\gamma}{1+\gamma}\epsilon=1-\frac{2+3\gamma}{1+\gamma}\epsilon,
\]
we have
\[
\left\Vert BD_{f,\lambda}^{-1}-\mathfrak{B}\mathfrak{D}_{f,\lambda}^{-1}\right\Vert \le t_{f}+\sqrt{1+t}\left(1-\sqrt{1-t}\right)\sqrt{\frac{\overline{\delta}_{w}+\lambda_{w}}{\lambda_{f}}}\le\left(\sqrt{\chi_{f}}+2\sqrt{\frac{\overline{\delta}_{w}+\lambda_{w}}{\lambda_{f}}}\right)t.
\]
With analogously defined $t_{w},$ we have $t_{w}\leq t\sqrt{\chi_{w}}$
where $\chi_{w}:=\max(1,M_{w}\overline{\delta}_{f}),$ and we have
with probability at least $1-\frac{3+2\gamma}{1+\gamma}\epsilon,$ 
\[
\left\Vert B^{\top}D_{w,\lambda}^{-1}-\mathcal{B^{\top}}\mathfrak{D}_{w,\lambda}^{-1}\right\Vert \le t_{w}+\sqrt{1+t}\left(1-\sqrt{1-t}\right)\sqrt{\frac{\overline{\delta}_{f}+\lambda_{f}}{\lambda_{w}}}\le\left(\sqrt{\chi_{w}}+2\sqrt{\frac{\overline{\delta}_{f}+\lambda_{f}}{\lambda_{w}}}\right)t
\]

\paragraph*{Bounds for $\left\Vert \mathfrak{B}\mathfrak{D}_{f,\lambda}^{-1}\right\Vert $
and $\left\Vert BD_{f,\lambda}^{-1}\right\Vert $, and for $\left\Vert \mathfrak{B^{\top}}\mathfrak{D}_{w,\lambda}^{-1}\right\Vert $
and $\left\Vert B^{\top}D_{w,\lambda}^{-1}\right\Vert $}

We have
\[
\left\Vert \mathfrak{B}\mathfrak{D}_{f,\lambda}^{-1}\right\Vert =\left\Vert \mathfrak{D}_{w,\lambda}^{1/2}\mathfrak{D}_{w,\lambda}^{-1/2}\mathfrak{B}\mathfrak{D}_{f,\lambda}^{-1/2}\mathfrak{D}_{f,\lambda}^{-1/2}\right\Vert \le\left\Vert \mathfrak{D}_{w,\lambda}^{1/2}\right\Vert \left\Vert \mathfrak{E}_{\lambda}\right\Vert \left\Vert \mathfrak{D}_{f,\lambda}^{-1/2}\right\Vert \le\sqrt{\frac{\overline{\delta}_{w}+\lambda_{w}}{\underline{\delta}_{f}+\lambda_{f}}},
\]
and
\begin{multline*}
\left\Vert BD_{f,\lambda}^{-1}\right\Vert =\left\Vert D_{w,\lambda}^{1/2}E_{\lambda}D_{f,\lambda}^{-1/2}\right\Vert =\left\Vert D_{w,\lambda}^{1/2}\mathfrak{D}_{w,\lambda}^{-1/2}\mathfrak{D}_{w,\lambda}^{1/2}E_{\lambda}D_{f,\lambda}^{-1/2}\right\Vert \\
\le\left\Vert D_{w,\lambda}^{1/2}\mathfrak{D}_{w,\lambda}^{-1/2}\right\Vert \left\Vert \mathfrak{D}_{w,\lambda}^{1/2}\right\Vert \left\Vert E_{\lambda}\right\Vert \left\Vert D_{f,\lambda}^{-1/2}\right\Vert \le\left\Vert D_{w,\lambda}^{1/2}\mathfrak{D}_{w,\lambda}^{-1/2}\right\Vert \sqrt{\frac{\overline{\delta}_{w}+\lambda_{w}}{\lambda_{f}}},
\end{multline*}
where $\Pr\left\{ \left\Vert D_{w,\lambda}^{1/2}\mathfrak{D}_{w,\lambda}^{-1/2}\right\Vert \ge\sqrt{1+t}\right\} \le\frac{1}{1+\gamma}\epsilon$.
Hence, with probability at least $1-\frac{1}{1+\gamma}\epsilon$,
\[
\left\Vert BD_{f,\lambda}^{-1}\right\Vert \le\sqrt{1+t}\sqrt{\frac{\overline{\delta}_{w}+\lambda_{w}}{\lambda_{f}}}\le2\sqrt{\frac{\overline{\delta}_{w}+\lambda_{w}}{\lambda_{f}}}.
\]

Similarly, we have $\left\Vert \mathfrak{B^{\top}}\mathfrak{D}_{w,\lambda}^{-1}\right\Vert \leq\sqrt{\frac{\overline{\delta}_{f}+\lambda_{f}}{\underline{\delta}_{w}+\lambda_{w}}}$,
and with probability at least $1-\frac{\gamma}{1+\gamma}\epsilon$,
$\left\Vert B^{\top}D_{w,\lambda}^{-1}\right\Vert \leq2\sqrt{\frac{\overline{\delta}_{f}+\lambda_{f}}{\lambda_{w}}}$.

\paragraph*{Main argument for Theorem \ref{thm:ConcentrationRidgeBias}}
The bias on $\widehat{\mu}$ is
\[
b_{\mu,\lambda}=\E\left(\widehat{\mu}-\mu\mid X,Z,\beta^{*}\right)=\widetilde{L}_{w,\lambda}^{-1}\left(-\lambda_{w}Z_{w}\mu^{*}+\lambda_{f}BD_{f,\lambda}^{-1}Z_{f}\phi^{*}\right).
\]
The ``population'' bias on $\mu$ is
\[
\mathfrak{b}_{\mu,\lambda}=\widetilde{\mathfrak{L}}_{w,\lambda}^{-1}\left(-\lambda_{w}Z_{w}\mu^{*}+\lambda_{f}\mathfrak{B}\mathfrak{D}_{f,\lambda}^{-1}Z_{f}\phi^{*}\right).
\]
Then (for the Euclidean norm),
\[
\left\Vert b_{\mu,\lambda}-\mathfrak{b}_{\mu,\lambda}\right\Vert \le\lambda_{w}\left\Vert \widetilde{L}_{w,\lambda}^{-1}-\widetilde{\mathfrak{L}}_{w,\lambda}^{-1}\right\Vert \left\Vert Z_{w}\mu^{*}\right\Vert +\lambda_{f}\left\Vert \widetilde{L}_{w,\lambda}^{-1}BD_{f,\lambda}^{-1}-\widetilde{\mathfrak{L}}_{w,\lambda}^{-1}\mathfrak{B}\mathfrak{D}_{f,\lambda}^{-1}\right\Vert \left\Vert Z_{f}\phi^{*}\right\Vert .
\]
We already have that
\begin{align*}
\Pr\left\{ \left\Vert \widetilde{L}_{w,\lambda}^{-1}-\widetilde{\mathfrak{L}}_{w,\lambda}^{-1}\right\Vert >\left(5+16\frac{\overline{\delta}_{w}+\lambda_{w}}{\underline{\delta}_{w}+\lambda_{w}}\right)\frac{\overline{\delta}_{w}+\lambda_{w}}{\lambda_{w}^{2}}t\right\}  & \le\frac{8+4\gamma}{1+\gamma}\epsilon,\\
\Pr\left\{ \left\Vert \widetilde{L}_{w,\lambda}^{-1}\right\Vert >2\frac{\overline{\delta}_{w}+\lambda_{w}}{\lambda_{w}^{2}}\right\}  & \le\frac{1}{1+\gamma}\epsilon,\\
\Pr\left\{ \left\Vert BD_{f,\lambda}^{-1}-\mathfrak{B}\mathfrak{D}_{f,\lambda}^{-1}\right\Vert >\left(\sqrt{\chi_{f}}+2\sqrt{\frac{\overline{\delta}_{w}+\lambda_{w}}{\lambda_{f}}}\right)t\right\}  & \le\frac{2+3\gamma}{1+\gamma}\epsilon,
\end{align*}
and
\[
\left\Vert \mathfrak{B}\mathfrak{D}_{f,\lambda}^{-1}\right\Vert =\left\Vert \mathfrak{D}_{w,\lambda}^{1/2}\mathfrak{E}\mathfrak{D}_{f,\lambda}^{-1/2}\right\Vert <\sqrt{\frac{\overline{\delta}_{w}+\lambda_{w}}{\underline{\delta}_{f}+\lambda_{f}}}.
\]
Furthermore,
\begin{align*}
\left\Vert \widetilde{L}_{w,\lambda}^{-1}BD_{f,\lambda}^{-1}-\widetilde{\mathfrak{L}}_{w,\lambda}^{-1}\mathfrak{B}\mathfrak{D}_{f,\lambda}^{-1}\right\Vert  & \le\left\Vert \left(\widetilde{L}_{w,\lambda}^{-1}-\widetilde{\mathfrak{L}}_{w,\lambda}^{-1}\right)\mathfrak{B}\mathfrak{D}_{f,\lambda}^{-1}+\widetilde{L}_{w,\lambda}^{-1}\left(BD_{f,\lambda}^{-1}-\mathfrak{B}\mathfrak{D}_{f,\lambda}^{-1}\right)\right\Vert \\
 & \le\left\Vert \widetilde{L}_{w,\lambda}^{-1}-\widetilde{\mathfrak{L}}_{w,\lambda}^{-1}\right\Vert \left\Vert \mathfrak{B}\mathfrak{D}_{f,\lambda}^{-1}\right\Vert +\left\Vert \widetilde{L}_{w,\lambda}^{-1}\right\Vert \left\Vert BD_{f,\lambda}^{-1}-\mathfrak{B}\mathfrak{D}_{f,\lambda}^{-1}\right\Vert .
\end{align*}
Hence, with probability at least $1-\frac{8+4\gamma}{1+\gamma}\epsilon-\frac{1}{1+\gamma}\epsilon-\frac{2+3\gamma}{1+\gamma}\epsilon=1-\frac{11+7\gamma}{1+\gamma}\epsilon$,
\begin{align*}
\left\Vert b_{\mu,\lambda}-\mathfrak{b}_{\mu,\lambda}\right\Vert  & \le\left(5+16\frac{\overline{\delta}_{w}+\lambda_{w}}{\underline{\delta}_{w}+\lambda_{w}}\right)\frac{\overline{\delta}_{w}+\lambda_{w}}{\lambda_{w}^{2}}t\left(\lambda_{w}\sqrt{n}\left\Vert \mu^{*}\right\Vert +\lambda_{f}\sqrt{\frac{\overline{\delta}_{w}+\lambda_{w}}{\underline{\delta}_{f}+\lambda_{f}}}\sqrt{p}\left\Vert \phi^{*}\right\Vert \right)\\
 & +2\lambda_{f}\frac{\overline{\delta}_{w}+\lambda_{w}}{\lambda_{w}^{2}}\left(\sqrt{\chi_{f}}+2\sqrt{\frac{\overline{\delta}_{w}+\lambda_{w}}{\lambda_{f}}}\right)t\sqrt{p}\left\Vert \phi^{*}\right\Vert .
\end{align*}
Finally, notice that
\[
\left\Vert \widetilde{\mathfrak{L}}_{w,\lambda}^{-1}\right\Vert =\left\Vert \mathfrak{D}_{w,\lambda}^{-1/2}\mathfrak{L}_{w,\lambda}^{-1}\mathfrak{D}_{w,\lambda}^{-1/2}\right\Vert \le\left\Vert \mathfrak{D}_{w,\lambda}^{-1}\right\Vert \left\Vert \mathfrak{L}_{w,\lambda}^{-1}\right\Vert =\frac{1}{\underline{\delta}_{w}+\lambda_{w}}\frac{\overline{\delta}_{w}+\lambda_{w}}{\lambda_{w}}.
\]
Hence,
\begin{align*}
\left\Vert \mathfrak{b}_{\mu,\lambda}\right\Vert  & \le\left\Vert \mathfrak{L}_{w,\lambda}^{-1}\right\Vert \left\Vert -\lambda_{w}Z_{w}\mu^{*}+\lambda_{f}\mathfrak{B}\mathfrak{D}_{f,\lambda}^{-1}Z_{f}\phi^{*}\right\Vert \\
 & \le\frac{1}{\underline{\delta}_{w}+\lambda_{w}}\frac{\overline{\delta}_{w}+\lambda_{w}}{\lambda_{w}}\left(\lambda_{w}\sqrt{n}\left\Vert \mu^{*}\right\Vert +\lambda_{f}\sqrt{\frac{\overline{\delta}_{w}+\lambda_{w}}{\underline{\delta}_{f}+\lambda_{f}}}\sqrt{p}\left\Vert \phi^{*}\right\Vert \right).
\end{align*}

Similarly, with probability at least $1-\frac{3+9\gamma}{1+\gamma}\epsilon-\frac{\gamma}{1+\gamma}\epsilon-\frac{3+2\gamma}{1+\gamma}\epsilon=1-\frac{6+12\gamma}{1+\gamma}\epsilon$,
we have
\begin{align*}
\left\Vert b_{\phi,\lambda}-\mathfrak{b}_{\phi,\lambda}\right\Vert  & \le\left(5+16\frac{\overline{\delta}_{f}+\lambda_{f}}{\underline{\delta}_{f}+\lambda_{f}}\right)\frac{\overline{\delta}_{f}+\lambda_{f}}{\lambda_{f}^{2}}t\left(\lambda_{f}\sqrt{p}\left\Vert \phi^{*}\right\Vert +\lambda_{w}\sqrt{\frac{\overline{\delta}_{f}+\lambda_{f}}{\underline{\delta}_{w}+\lambda_{w}}}\sqrt{n}\left\Vert \mu^{*}\right\Vert \right)\\
 & +2\lambda_{w}\frac{\overline{\delta}_{f}+\lambda_{f}}{\lambda_{f}^{2}}\left(\sqrt{\chi_{w}}+2\sqrt{\frac{\overline{\delta}_{f}+\lambda_{f}}{\lambda_{w}}}\right)t\sqrt{n}\left\Vert \mu^{*}\right\Vert.
\end{align*}
Moreover,
\[
\left\Vert \widetilde{\mathfrak{L}}_{f,\lambda}^{-1}\right\Vert =\left\Vert \mathfrak{D}_{f,\lambda}^{-1/2}\mathfrak{L}_{f,\lambda}^{-1}\mathfrak{D}_{f,\lambda}^{-1/2}\right\Vert \le\left\Vert \mathfrak{D}_{f,\lambda}^{-1}\right\Vert \left\Vert \mathfrak{L}_{f,\lambda}^{-1}\right\Vert =\frac{1}{\underline{\delta}_{f}+\lambda_{f}}\frac{\overline{\delta}_{f}+\lambda_{f}}{\lambda_{f}}.
\]
Hence, we also have
\begin{align*}
\left\Vert \mathfrak{b}_{\phi,\lambda}\right\Vert  & \le\left\Vert \mathfrak{L}_{f,\lambda}^{-1}\right\Vert \left\Vert -\lambda_{f}Z_{f}\phi^{*}+\lambda_{w}\mathfrak{B^{\top}}\mathfrak{D}_{w,\lambda}^{-1}Z_{w}\mu^{*}\right\Vert \\
 & \le\frac{1}{\underline{\delta}_{f}+\lambda_{f}}\frac{\overline{\delta}_{f}+\lambda_{f}}{\lambda_{f}}\left(\lambda_{f}\sqrt{p}\left\Vert \phi^{*}\right\Vert +\lambda_{w}\sqrt{\frac{\overline{\delta}_{f}+\lambda_{f}}{\underline{\delta}_{w}+\lambda_{w}}}\sqrt{n}\left\Vert \mu^{*}\right\Vert \right).
\end{align*}

\subsection{Proof of Theorem \ref{thm:ConcentrationRidgeVariance}}

The variance of $\widehat{\mu}-\mu$ is
\[
V_{w,\lambda}=\sigma^{2}\widetilde{L}_{w,\lambda}^{-1}+\left(\lambda_{w}^{2}\sigma_{w}^{2}-\lambda_{w}\sigma^{2}\right)\widetilde{L}_{w,\lambda}^{-2}+\left(\lambda_{f}^{2}\sigma_{f}^{2}-\lambda_{f}\sigma^{2}\right)\widetilde{L}_{w,\lambda}^{-1}BD_{f,\lambda}^{-2}B^{\top}\widetilde{L}_{w,\lambda}^{-1}.
\]
The ``population'' variance is
\[
\mathfrak{V}_{w,\lambda}=\sigma^{2}\widetilde{\mathfrak{L}}_{w,\lambda}^{-1}+\left(\lambda_{w}^{2}\sigma_{w}^{2}-\lambda_{w}\sigma^{2}\right)\widetilde{\mathfrak{L}}_{w,\lambda}^{-2}+\left(\lambda_{f}^{2}\sigma_{f}^{2}-\lambda_{f}\sigma^{2}\right)\widetilde{\mathfrak{L}}_{w,\lambda}^{-1}\mathfrak{B}\mathfrak{D}_{f,\lambda}^{-2}\mathfrak{B}^{\top}\widetilde{\mathfrak{L}}_{w,\lambda}^{-1}.
\]
Taking the difference we get
\begin{multline*}
\left\Vert V_{w,\lambda}-\mathfrak{V}_{w,\lambda}\right\Vert \le\sigma^{2}\left\Vert \widetilde{L}_{w,\lambda}^{-1}-\widetilde{\mathfrak{L}}_{w,\lambda}^{-1}\right\Vert +\left|\lambda_{w}^{2}\sigma_{w}^{2}-\lambda_{w}\sigma^{2}\right|\left\Vert \widetilde{L}_{w,\lambda}^{-2}-\widetilde{\mathfrak{L}}_{w,\lambda}^{-2}\right\Vert \\
+\left|\lambda_{f}^{2}\sigma_{f}^{2}-\lambda_{f}\sigma^{2}\right|\left\Vert \widetilde{L}_{w,\lambda}^{-1}BD_{f,\lambda}^{-2}B^{\top}\widetilde{L}_{w,\lambda}^{-1}-\widetilde{\mathfrak{L}}_{w,\lambda}^{-1}\mathfrak{B}\mathfrak{D}_{f,\lambda}^{-2}\mathfrak{B}^{\top}\widetilde{\mathfrak{L}}_{w,\lambda}^{-1}\right\Vert .
\end{multline*}
We already know that 
\begin{align*}
\Pr\left\{ \left\Vert \widetilde{L}_{w,\lambda}^{-1}-\widetilde{\mathfrak{L}}_{w,\lambda}^{-1}\right\Vert >\left(5+16\frac{\overline{\delta}_{w}+\lambda_{w}}{\underline{\delta}_{w}+\lambda_{w}}\right)\frac{\overline{\delta}_{w}+\lambda_{w}}{\lambda_{w}^{2}}t\right\}  & \le\frac{8+4\gamma}{1+\gamma}\epsilon,\\
\Pr\left\{ \left\Vert \widetilde{L}_{w,\lambda}^{-1}\right\Vert >2\frac{\overline{\delta}_{w}+\lambda_{w}}{\lambda_{w}^{2}}\right\}  & \le\frac{1}{1+\gamma}\epsilon,\\
\Pr\left\{ \left\Vert BD_{f,\lambda}^{-1}-\mathfrak{B}\mathfrak{D}_{f,\lambda}^{-1}\right\Vert >\left(\sqrt{\chi_{f}}+2\sqrt{\frac{\overline{\delta}_{w}+\lambda_{w}}{\lambda_{f}}}\right)t\right\}  & \le\frac{2+3\gamma}{1+\gamma}\epsilon,\\
\Pr\left\{ \left\Vert BD_{f,\lambda}^{-1}\right\Vert >2\sqrt{\frac{\overline{\delta}_{w}+\lambda_{w}}{\lambda_{f}}}\right\}  & \le\frac{1}{1+\gamma}\epsilon,
\end{align*}
and that
\begin{align*}
\left\Vert \widetilde{\mathfrak{L}}_{w,\lambda}^{-1}\right\Vert  & =\left\Vert \mathfrak{D}_{w,\lambda}^{-1/2}\mathfrak{L}_{w,\lambda}^{-1}\mathfrak{D}_{w,\lambda}^{-1/2}\right\Vert \le\left\Vert \mathfrak{D}_{w,\lambda}^{-1}\right\Vert \left\Vert \mathfrak{L}_{w,\lambda}^{-1}\right\Vert =\frac{1}{\underline{\delta}_{w}+\lambda_{w}}\frac{\overline{\delta}_{w}+\lambda_{w}}{\lambda_{w}},\\
\left\Vert \mathfrak{B}\mathfrak{D}_{f,\lambda}^{-1}\right\Vert  & =\left\Vert \mathfrak{D}_{w,\lambda}^{1/2}\mathfrak{E}\mathfrak{D}_{f,\lambda}^{-1/2}\right\Vert <\sqrt{\frac{\overline{\delta}_{w}+\lambda_{w}}{\underline{\delta}_{f}+\lambda_{f}}}.
\end{align*}
It follows that, with probability at least 
\[
1-\frac{8+4\gamma}{1+\gamma}\epsilon-\frac{1}{1+\gamma}\epsilon-\frac{2+3\gamma}{1+\gamma}\epsilon-\frac{1}{1+\gamma}\epsilon=1-\frac{12+7\gamma}{1+\gamma}\epsilon,
\]
we have the following inequalities. First,
\begin{multline*}
\left\Vert \widetilde{L}_{w,\lambda}^{-1}BD_{f,\lambda}^{-1}-\widetilde{\mathfrak{L}}_{w,\lambda}^{-1}\mathfrak{B}\mathfrak{D}_{f,\lambda}^{-1}\right\Vert \le\left\Vert \widetilde{L}_{w,\lambda}^{-1}-\widetilde{\mathfrak{L}}_{w,\lambda}^{-1}\right\Vert \left\Vert \mathfrak{B}\mathfrak{D}_{f,\lambda}^{-1}\right\Vert +\left\Vert \widetilde{L}_{w,\lambda}^{-1}\right\Vert \left\Vert BD_{f,\lambda}^{-1}-\mathfrak{B}\mathfrak{D}_{f,\lambda}^{-1}\right\Vert \\
\le\left(5+16\frac{\overline{\delta}_{w}+\lambda_{w}}{\underline{\delta}_{w}+\lambda_{w}}\right)\frac{\overline{\delta}_{w}+\lambda_{w}}{\lambda_{w}^{2}}t\sqrt{\frac{\overline{\delta}_{w}+\lambda_{w}}{\underline{\delta}_{f}+\lambda_{f}}}+2\frac{\overline{\delta}_{w}+\lambda_{w}}{\lambda_{w}^{2}}\left(\sqrt{\chi_{f}}+2\sqrt{\frac{\overline{\delta}_{w}+\lambda_{w}}{\lambda_{f}}}\right)t.
\end{multline*}
Second,
\begin{align*}
\left\Vert \widetilde{L}_{w,\lambda}^{-2}-\widetilde{\mathfrak{L}}_{w,\lambda}^{-2}\right\Vert  & \le\left(\left\Vert \widetilde{L}_{w,\lambda}^{-1}\right\Vert +\left\Vert \widetilde{\mathfrak{L}}_{w,\lambda}^{-1}\right\Vert \right)\left\Vert \widetilde{L}_{w,\lambda}^{-1}-\widetilde{\mathfrak{L}}_{w,\lambda}^{-1}\right\Vert \\
 & \le\frac{\overline{\delta}_{w}+\lambda_{w}}{\lambda_{w}}\left(\frac{2}{\lambda_{w}}+\frac{1}{\underline{\delta}_{w}+\lambda_{w}}\right)\left(5+16\frac{\overline{\delta}_{w}+\lambda_{w}}{\underline{\delta}_{w}+\lambda_{w}}\right)\frac{\overline{\delta}_{w}+\lambda_{w}}{\lambda_{w}^{2}}t.
\end{align*}
Third,
\begin{multline*}
\left\Vert \widetilde{L}_{w,\lambda}^{-1}BD_{f,\lambda}^{-2}B^{\top}\widetilde{L}_{w,\lambda}^{-1}-\widetilde{\mathfrak{L}}_{w,\lambda}^{-1}\mathfrak{B}\mathfrak{D}_{f,\lambda}^{-2}\mathfrak{B}^{\top}\widetilde{\mathfrak{L}}_{w,\lambda}^{-1}\right\Vert \\
\le\left(\left\Vert \widetilde{L}_{w,\lambda}^{-1}BD_{f,\lambda}^{-1}\right\Vert +\left\Vert \widetilde{\mathfrak{L}}_{w,\lambda}^{-1}\mathfrak{B}\mathfrak{D}_{f,\lambda}^{-1}\right\Vert \right)\left\Vert \widetilde{L}_{w,\lambda}^{-1}BD_{f,\lambda}^{-1}-\widetilde{\mathfrak{L}}_{w,\lambda}^{-1}\mathfrak{B}\mathfrak{D}_{f,\lambda}^{-1}\right\Vert \\
\le\left(\left\Vert \widetilde{L}_{w,\lambda}^{-1}\right\Vert \left\Vert BD_{f,\lambda}^{-1}\right\Vert +\left\Vert \widetilde{\mathfrak{L}}_{w,\lambda}^{-1}\right\Vert \left\Vert \mathfrak{B}\mathfrak{D}_{f,\lambda}^{-1}\right\Vert \right)\left\Vert \widetilde{L}_{w,\lambda}^{-1}BD_{f,\lambda}^{-1}-\widetilde{\mathfrak{L}}_{w,\lambda}^{-1}\mathfrak{B}\mathfrak{D}_{f,\lambda}^{-1}\right\Vert \\
\le\left(2\frac{\overline{\delta}_{w}+\lambda_{w}}{\lambda_{w}^{2}}2\sqrt{\frac{\overline{\delta}_{w}+\lambda_{w}}{\lambda_{f}}}+\frac{1}{\underline{\delta}_{w}+\lambda_{w}}\frac{\overline{\delta}_{w}+\lambda_{w}}{\lambda_{w}}\sqrt{\frac{\overline{\delta}_{w}+\lambda_{w}}{\underline{\delta}_{f}+\lambda_{f}}}\right)\\
\times\left[\left(5+16\frac{\overline{\delta}_{w}+\lambda_{w}}{\underline{\delta}_{w}+\lambda_{w}}\right)\frac{\overline{\delta}_{w}+\lambda_{w}}{\lambda_{w}^{2}}\sqrt{\frac{\overline{\delta}_{w}+\lambda_{w}}{\underline{\delta}_{f}+\lambda_{f}}}+2\frac{\overline{\delta}_{w}+\lambda_{w}}{\lambda_{w}^{2}}\left(\sqrt{\chi_{f}}+2\sqrt{\frac{\overline{\delta}_{w}+\lambda_{w}}{\lambda_{f}}}\right)\right]t.
\end{multline*}
Finally,
\begin{multline*}
\left\Vert V_{w,\lambda}-\mathfrak{V}_{w,\lambda}\right\Vert \le\sigma^{2}\left(5+16\frac{\overline{\delta}_{w}+\lambda_{w}}{\underline{\delta}_{w}+\lambda_{w}}\right)\frac{\overline{\delta}_{w}+\lambda_{w}}{\lambda_{w}^{2}}t\\
+\left|\lambda_{w}^{2}\sigma_{w}^{2}-\lambda_{w}\sigma^{2}\right|\left(\frac{2}{\lambda_{w}}+\frac{1}{\underline{\delta}_{w}+\lambda_{w}}\right)\left(5+16\frac{\overline{\delta}_{w}+\lambda_{w}}{\underline{\delta}_{w}+\lambda_{w}}\right)\frac{\left(\overline{\delta}_{w}+\lambda_{w}\right)^{2}}{\lambda_{w}^{3}}t\\
+\left|\lambda_{f}^{2}\sigma_{f}^{2}-\lambda_{f}\sigma^{2}\right|\left(4\frac{1}{\lambda_{w}^ {}}\sqrt{\frac{1}{\lambda_{f}}}+\frac{1}{\underline{\delta}_{w}+\lambda_{w}}\sqrt{\frac{1}{\underline{\delta}_{f}+\lambda_{f}}}\right)\frac{\left(\overline{\delta}_{w}+\lambda_{w}\right)^{5/2}}{\lambda_{w}^{3}}\\
\times\left[\left(5+16\frac{\overline{\delta}_{w}+\lambda_{w}}{\underline{\delta}_{w}+\lambda_{w}}\right)\sqrt{\frac{\overline{\delta}_{w}+\lambda_{w}}{\underline{\delta}_{f}+\lambda_{f}}}+2\left(\sqrt{\chi_{f}}+2\sqrt{\frac{\overline{\delta}_{w}+\lambda_{w}}{\lambda_{f}}}\right)\right]t.
\end{multline*}

Similarly, with probability at least 
\[
1-\frac{3+9\gamma}{1+\gamma}\epsilon-\frac{\gamma}{1+\gamma}\epsilon-\frac{3+2\gamma}{1+\gamma}\epsilon-\frac{\gamma}{1+\gamma}\epsilon=1-\frac{5+13\gamma}{1+\gamma}\epsilon,
\]
we have
\begin{multline*}
\left\Vert V_{f,\lambda}-\mathfrak{V}_{f,\lambda}\right\Vert \le\sigma^{2}\left(5+16\frac{\overline{\delta}_{f}+\lambda_{f}}{\underline{\delta}_{f}+\lambda_{f}}\right)\frac{\overline{\delta}_{f}+\lambda_{f}}{\lambda_{f}^{2}}t\\
+\left|\lambda_{f}^{2}\sigma_{f}^{2}-\lambda_{f}\sigma^{2}\right|\left(\frac{2}{\lambda_{f}}+\frac{1}{\underline{\delta}_{f}+\lambda_{f}}\right)\left(5+16\frac{\overline{\delta}_{f}+\lambda_{f}}{\underline{\delta}_{f}+\lambda_{f}}\right)\frac{\left(\overline{\delta}_{f}+\lambda_{f}\right)^{2}}{\lambda_{f}^{3}}t\\
+\left|\lambda_{w}^{2}\sigma_{w}^{2}-\lambda_{w}\sigma^{2}\right|\left(4\frac{1}{\lambda_{f}}\sqrt{\frac{1}{\lambda_{w}}}+\frac{1}{\underline{\delta}_{f}+\lambda_{f}}\sqrt{\frac{1}{\underline{\delta}_{w}+\lambda_{w}}}\right)\frac{\left(\overline{\delta}_{f}+\lambda_{f}\right)^{5/2}}{\lambda_{f}^{3}}\\
\times\left[\left(5+16\frac{\overline{\delta}_{f}+\lambda_{f}}{\underline{\delta}_{f}+\lambda_{f}}\right)\sqrt{\frac{\overline{\delta}_{f}+\lambda_{f}}{\underline{\delta}_{w}+\lambda_{w}}}+2\left(\sqrt{\chi_{w}}+2\sqrt{\frac{\overline{\delta}_{f}+\lambda_{f}}{\lambda_{w}}}\right)\right]t.
\end{multline*}
This ends the proofs.

\end{appendix}

\begin{acks}[Acknowledgments]
The authors would like to thank the anonymous referees, an Associate
Editor and the Editor for their constructive comments that improved the
quality of this paper. A first version of this paper was presented at the Econometric Study Group, Bristol 2023. A more elaborate version was presented at Econometrics in Rio 2024 and at the conference in honor of Thierry Magnac, May 2025, where Koen Jochmans made useful comments. We thank seminar participants at the University of Chicago (May 2025) for helpful comments and suggestions. 
\end{acks}

\begin{funding}
The second author  acknowledges support from the European Research Council (grant reference ERC-2020-ADG-101018130).
\end{funding}



\bibliographystyle{ecta-fullname}
\bibliography{RidgeAKM}

\end{document}